\documentclass[11pt]{article}
\usepackage[margin=1.1in]{geometry}
\usepackage{amsmath,amssymb,amsthm}
\usepackage{booktabs}
\usepackage[colorlinks=true,linkcolor=blue,citecolor=blue,urlcolor=blue]{hyperref}
\usepackage{setspace}
\usepackage[round]{natbib}
\usepackage{nameref}
\newtheorem{theorem}{Theorem}

\newtheorem{lemma}{Lemma}
\newtheorem{corollary}{Corollary}
\newtheorem{definition}{Definition}
\newtheorem{assumption}{Assumption}
\theoremstyle{remark}

\newtheorem{example}{Example}
\newcommand{\cR}{\mathcal{R}}
\newcommand{\cL}{\mathcal{L}}
\newcommand{\cD}{\mathcal{D}}

\newcommand{\cG}{\mathcal{G}}

\newcommand{\cE}{\mathcal{E}}
\newcommand{\cT}{\mathcal{T}}
\newcommand{\rs}{\rho^{*}}
\newcommand{\1}{\mathbf{1}}
\newcommand{\cS}{\mathcal{S}}
\title{Random Utility Model with Endogenously Assigned Menus}
\author{Kai Wang\thanks{University of Bath, kw905@bath.ac.uk.}}
\date{8 July 2026}

\begin{document}
\maketitle

\begin{abstract}
There is a largely overlooked assumption underlying stochastic choice theory: menus are assigned exogenously, as if by a hidden randomized controlled trial (RCT). This assumption is not innocuous, because in many real-world settings menus are assigned endogenously according to the decision makers' (DMs) preferences. This paper studies Random Utility Model (RUM) under menu endogeneity and shows that any seemingly anomalous choice behavior can be generated by a population of rational DMs with heterogeneous preferences facing endogenously assigned menus. To address this problem, I propose a new causal estimand: the probability that an alternative would be chosen from a given menu if that menu were presented to the entire population. I then characterize sharp bounds on this estimand using observed menu and choice frequencies. In addition, I show that when choices are observed across multiple markets with different preference distributions and a common menu-assignment rule, the causal estimands from those markets and the menu assignment rule are uniquely identified.

Keywords: random utility model, counterfactual choice, endogenous menu assignment.
\end{abstract}
\section{Introduction}\label{sec:intro}

In many real-world settings menus are assigned endogenously according to the DMs' preferences. In this paper, I study RUM with endogenously assigned menus. I propose the causal estimand that should be used for analysis: the counterfactual choice probabilities if menus were assigned exogenously. I then show how an observer of choice and menu frequencies can identify the sharp bounds on both the causal estimand and the menu assignment rule. I further show that if choices and menu frequencies are observed across several markets with different preference distributions, both objects can be uniquely identified under the assumption that the menu assignment is outcome-based and common cross market.

Stochastic choice theory takes as primitive a choice rule $\rho$, where $\rho(a \mid A)$ denotes the probability that alternative $a$ is chosen from menu $A$. The central exercise of the literature, from Block and Marschak (1960) and Falmagne (1978) through Gul and Pesendorfer (2006) and the modern literature on bounded rationality and limited attention (Manzini and Mariotti, 2014; Cattaneo, Ma, Masatlioglu, and Suleymanov, 2020), is to ask which restrictions on $\rho$ characterize which behavioral models.\footnote{The fundamental axioms, e.g., Regularity, the Block-Marschak Inequalities, Independence of Irrelevant Alternatives (IIA), are all conditions that compare choice frequencies across menus.}

Every such cross-menu comparison, however, carries an implicit assumption: the DMs facing menu $A$ are drawn from the same population, in distribution, as those facing menu $B\neq{A}$. In the language of causal inference, menus must be assigned exogenously. Thus, much of stochastic choice theory relies on a hidden RCT: the analyst compares choices across menus as if menu assignment were exogenously randomized. However, in many modern empirical settings, e.g., scanner panels, insurance-plan choices, online platforms, and retail data, menu assignment is endogenous. For example, a food-delivery platform's algorithm shows users restaurants similar to their past orders; an insurance broker presents plans they expect a household to like; retail assortments respond to local demand. Perhaps surprisingly, this problem has been largely ignored in the economic literature. This paper formalizes this problem and develops identification strategies when menus are endogenously assigned. I now motivate my theory with a simple example.

\begin{example}\label{sec:example}
A delivery platform serves a city with two restaurants: \emph{Aki}, a sushi place ($a$), and \emph{Bella Napoli}, a pizzeria ($b$). A third restaurant, \emph{Chiyo} ($c$), has just joined the platform, and it serves high-end sushi that no one orders online. The population splits equally into two rational types: sushi lovers (type $\alpha$, mass $50\%$, ranking $a \succ b \succ c$) and pizza lovers (type $\beta$, mass $50\%$, ranking $b \succ a \succ c$). 

The platform uses a ``content-based'' recommender system, which solves the cold-start problem by including $c$ into the feeds of the users who are most likely to order it.\footnote{Lops, de Gemmis, and Semeraro (2011) provide an overview of content-based recommendation systems.} Users whose order history says ``sushi'' get the augmented feed $\left\{ a,b,c\right\} $; users whose history says ``pizza'' keep the standard feed. Allowing for the algorithm's occasional misclassification and exploration, type $\alpha$ sees the augmented feed $\{a,b,c\}$ with probability $0.9$ and the standard feed $\{a,b\}$ with probability $0.1$; type $\beta$ sees the augmented feed $\{a,b,c\}$ with probability $0.1$ and the standard feed $\{a,b\}$ with probability $0.9$.

An analyst observes the standard choice data: feed shown (menu), and restaurant ordered (choice). By Bayes' rule,
\[
\rho(a \mid \{a,b\}) = 0.1, \qquad \rho(a \mid \{a,b,c\}) = 0.9.
\]
\end{example}

The pattern is a textbook asymmetric-dominance finding. Alternative $c$ resembles a dominated decoy for $a$: $c$ itself is never chosen, and $a$'s share jumps ninefold when $c$ is present. The analyst may conclude that placing an overpriced sushi option $c$ next to $a$ boosts the choice probability of $a$. The truth, however, is that $c$'s presence in a menu does not \emph{attract} users to $a$, as the agents in this example are rational DMs. The decoy-like choice probabilities happen because menus are assigned endogenously. I show that this example generalizes: any choice data can be generated by a population of perfectly rational DMs with heterogeneous preferences facing endogenously assigned menus. Thus, even choice patterns that appear highly anomalous may be artifacts of selection rather than evidence against rationality.

This result shows that any existing method that studies $\rho$ without accounting for menu endogeneity may be misleading, as studying $\rho$ confounds how preferences rank alternatives with which preference types end up facing which menus. A natural question is: what estimand should we study when the menus are endogenously assigned? In this paper, I propose a novel estimand that restores a causal interpretation. This causal estimand recovers the choice probability that would obtain under exogenous menu assignment. Let $\succ$ be any strict linear order on $X$; let $\cL$ denote the set of all strict linear orders (preferences), and let $\nu \in \Delta(\cL)$ be a distribution over preferences. I claim that the causal estimands are the \emph{counterfactual choice probabilities}:
\begin{equation}\label{eq:rhostar}
\rs(a \mid A) = \nu\big(\{\succ \in \cL : a \succ b \;\; \forall b \in A \setminus \{a\}\}\big),
\end{equation}
that is, the share that would choose $a$ if menu $A$ were assigned to the whole population.

Throughout the paper, the analyst needs to observe the conditional choice probability $\rho$ as well as the distribution of menus $\mu$. The literature on stochastic choice theory studies $\rho$ alone, treating $\mu$ as irrelevant exactly because of the hidden assumption that the menus are assigned exogenously. If menus are assigned independently of preference types, then the population facing any menu is representative of the overall population, and $\mu$ conveys no information regarding the menu assignment. For my purposes, $\mu$ is indispensable. The intuition is that $\mu(A)$ is the \emph{weight of evidence} that menu $A$ contributes about population preferences. For example, observing that $90\%$ of respondents choose $a$ over $b$ when shown $\{a,b,c\}$ is evidence that at least $90\%$ of the population prefers $a$ to $b$, but only if those respondents are representative of the population. If menu $\{a,b,c\}$
were shown to only a small fraction of the population, the high choice rate of $a$ could reflect extreme selection into that menu rather than population-wide preferences.

What is more, observing $\mu$ imposes almost no data-collection burden: it is recorded alongside $\rho$ in almost every applied choice dataset, but discarded at the
analysis stage when researchers condition on the menu rather than modelling it
jointly. For example, many scanner datasets (e.g., IRI, Nielsen) log both the product assortment on the shelf and the item purchased.\footnote{Conlon and Mortimer (2013) use vending-machine stock-outs where the menu is directly observed.} In other words, this requirement is already met in standard datasets that log both the offered alternatives (menu) and the chosen one (choice).

My first main contribution is to propose a method to identify the sharp bounds on $\rs$ and the menu assignment rule from $\rho$ and $\mu$ only. Given the observed menu frequencies $\mu$ and conditional choice probabilities $\rho$, I characterize the sharp identified set for $\rs$: the smallest set of counterfactual choice probabilities consistent with the data and rationality.
To establish the identification results, I translate the problem into a
transportation problem: the masses of the preference types are supplies,
the observed menu-choice joint frequencies are demands, and mass may be
transported from a type to an observation only if the type would make that
choice from that menu. The key step is employing the Feasibility Theorem
of Gale (1957) from transportation theory, which provides a necessary and
sufficient condition for the existence of a transport plan between given
supplies and demands. This tool delivers the sharp identified set for the
preference distribution, and with it sharp bounds on the counterfactual
choice probabilities $\rs$ and on the menu-assignment probabilities.

Notably, the bounds are well defined even for menus that \emph{never} appear in
the data and for alternatives that are never chosen: types that would
select an unobserved alternative are partially revealed by their choices
elsewhere, so the data remain informative about counterfactuals the
analyst never directly observes. In Appendix~\ref{example:sharp}, I illustrate how to derive the sharp bounds for $\rs$ and the menu assignment rule in the motivating example.\footnote{In this motivating example, these bounds are remarkably narrow. The counterfactual choice probability of $a$ from $\left\{ a,b\right\} $ (i.e., $\rho^{*}\left(a\mid\left\{ a,b\right\} \right)$) is point identified
at $0.5$, and from $\left\{ a,b,c\right\}$ (i.e., $\rho^{*}\left(a\mid\left\{ a,b,c\right\} \right)$) is identified up to the
interval $\left[0.45,\,0.5\right]$, despite conditional choice probabilities $\rho$
of $0.1$ and $0.9$. The apparent decoy effect is thus bounded away from a
single market: adding $c$ to the feed cannot raise the counterfactual
share of $a$ at all.}

I also develop a closed-form expression of the bounds using a classical theorem of Shapley (1971) on supermodular
set functions: both endpoints reduce to finite sums of observed choice frequencies, so the sharp bounds
can be computed from any dataset by summation alone. In Appendix~\ref{example:closed}, I show how one can identify the closed-form sharp bounds in the motivating example as well.

My second main contribution is to show that unique identification of the causal estimand $\rs$ and the menu assignment rule is obtained under an intuitive restriction on the menu-assignment process. To do that, the analyst needs to observe choice data from different markets, each of which contains a different distribution of preference types. This type of data set is common in practice: an insurance company may operate across regions with different customer populations, a chain supermarket may face different local demand across stores, and a food-delivery platform may serve neighborhoods or cities with different tastes. 

I also assume that the probability of assigning a menu conditional on the preference type is determined by the alternative that the preference type is predicted to choose from that menu, and that this menu assignment rule is common across markets. This assumption is plausible in many applications. Platforms, brokers, and retailers often target menus using predictions of which alternative a consumer is choosing. Indeed, a menu designer who seeks to induce, facilitate, or monetize particular choices has a natural reason to condition assignment on the predicted chosen alternative rather than on the consumer's entire latent preference ranking.\footnote{See, e.g., Covington, Adams, and Sargin (2016), Knott, Hayes, and Neslin (2002).} The common-assignment component is also natural when the same designer operates across markets. If a menu-assignment rule is operationally successful in one market, the designer has incentives to deploy the same rule elsewhere.\footnote{The assumption would fail if the designer deliberately used different assignment rules across markets, for example as part of an A/B test or a market-specific experimentation policy.} When the same assignment rule operates across several markets with different preference distributions, I show that the preference variation across markets is powerful for identification: under a simple rank condition, I show that the observed menu and choice frequencies across markets uniquely identify both the counterfactual choice probabilities $\rho^*$ and the menu assignment rule. In Appendix \ref{example:unique}, I also how one can use this result to identify uniquely $\rs$ and the menu assignment rule in the motivating example.

\section{Literature Review}

The causal estimand I study is related to Manski (2007), which also considers partial identification of counterfactual choice probabilities. In his framework, observed choice probabilities partially identify the distribution of preference types in the population, and this partial knowledge in turn yields bounds on choice probabilities in unobserved menus. The key difference lies in where the identification problem originates. In Manski (2007), the DMs face exogenously assigned menus, so the observed choice probabilities are directly informative about population preferences, and the difficulty is extrapolation to menus never observed in the data. In my setting, by contrast, even the observed conditional choice probabilities are contaminated by selection into menus, and the counterfactual of interest is the assignment of an \emph{observed} menu to the entire population. This shift also changes the data requirement: identification here rests on using the menu distribution $\mu$ with the conditional choice rule $\rho$, whereas $\mu$ plays no role in Manski's (2007) analysis.

A closely related literature studies choice data in which part
of the environment is unobserved. Aguiar and Kashaev (2025) recover the
joint distribution of choice sets and preferences from a cross-section of repeated choices. Agarwal and Somaini (2025) identify demand under
latent choice constraints using two sets of instruments, shifters of
preferences and shifters of menus. Liu and Luo (2025) estimate
demand when a monopolist facing price rigidity adjusts unobserved
assortments. Closest in technique, Kono, Saito, and Sandroni (2025)
characterize the testable content of the RUM when the choice frequencies of
some alternatives are unobservable, deriving---as I do---a system of linear
inequalities from a flow-feasibility theorem, building on the network-flow approach of Fiorini (2004). What separates my exercise
from all of the work above is that nothing in my data is missing: menus and
choice frequencies are fully observed, and the difficulty is instead that
endogenous assignment contaminates the interpretation of the conditional
frequencies. The contrast is sharpest with Kono, Saito, and Sandroni
(2025): under their relaxation of the classical data assumptions,
rationality retains testable content, which they characterize exactly;
under mine it retains none (Theorem~\ref{thm:anything}), and the analysis
turns from testing the model to bounding counterfactuals and the assignment
rule.

The recent study of Tomlinson, Ugander, and Benson (2021) in computer science works on this endogeneity issue. They work with parametric discrete-choice models and use observed covariates of the DM to analyze choices when menus are endogenous. They also exploit structured menu-assignment mechanisms to improve out-of-sample prediction. My paper addresses this question in a very different way: my analysis is nonparametric, and my identification does not rely on observing covariates of the DM; the analyst needs only the choice probabilities and menu frequencies. The identification of the sharp bounds for the estimand imposes no structure whatsoever on the assignment rule, and where I do impose structure to obtain point identification, the restriction is placed directly on the assignment rule itself (that steering may depend on the DM's preference only through the predicted choice from the menu shown), which is a statement about how the rule uses the latent preference, not about observables.

A complementary literature on limited attention endogenizes not the menu but the \emph{consideration set}, the subset of the menu that the DM actually evaluates. Masatlioglu, Nakajima, and Ozbay (2012), Manzini and Mariotti (2014), Caplin and Dean (2015), Cattaneo, Ma, Masatlioglu, and Suleymanov (2020), Gibbard (2021), and Aguiar, Boccardi, Kashaev, and Kim (2023) characterize such models axiomatically and show how the choice probability $\rho$ can be used to identify preference and consideration set. Closest in spirit within this literature is Dardanoni, Manzini, Mariotti, and Tyson (2020), whose focus, like mine, is identification: they show that aggregate choice shares, observed jointly across three ``occasions,'' uniquely pin down the population distribution of cognitive characteristics. Dardanoni, Manzini, Mariotti, Petri, and Tyson (2023) extend this program to mixture choice data, the joint distribution of choices of the same agents across a collection of menus, and show that such data simultaneously identify the population's preferences and cognition. My unique identification result requires no such joint choice share data: choice frequencies together with menu frequencies from different markets suffice for unique identification. On the econometric side, Abaluck and Adams-Prassl (2021) develop a general framework in which product characteristics are observable, unlike in my setting, and exploit asymmetries in cross-characteristic choice probability responses to identify consideration sets. For choice under risk, Barseghyan, Coughlin, Molinari, and Teitelbaum (2021) study preferences and attention in a general model with minimal assumptions on the consideration-set formation process, targeting partial identification of its components. Crawford, Griffith, and Iaria (2021) propose an identification strategy based on reducing the menu to a ``sufficient set'' of alternatives that are certain to be considered. Gaynor, Propper, and Seiler (2016) exploit institutional changes to identify consideration sets in hospital choice, while Honka, Horta\c{c}su, and Vitorino (2017) treat consideration sets as the outcome of a costly search process. All of these studies, however, assume implicitly that the menu itself is exogenously assigned, and endogeneity enters only through the mapping from the menu to the consideration set. This paper studies a completely different problem, in which the DMs evaluate every alternative from the menu, but their menus are assigned endogenously.

Another related literature studies the empirical content and identification of preferences within RUM. The classical work of Block and Marschak (1960), Falmagne (1978), and McFadden and Richter (1991) characterizes when a stochastic choice rule admits a random utility representation; related work, including Barberá and Pattanaik (1986) and Kitamura and Stoye (2018), studies stochastic rationalizability and its testable implications. A more recent strand asks when the underlying distribution over preferences is uniquely determined by the stochastic choice rule. Turansick (2022) characterizes when a random utility representation is unique; Apesteguia, Ballester, and Lu (2017) obtain uniqueness in single crossing random utility models; Apesteguia and Ballester (2023) study ordered type restrictions on restricted menu domains; and Chambers and Turansick (2025) and Caradonna and Turansick (2026) study the limits of identification and the restrictions that restore it. This paper abstracts away from this preference identification problem. Its identification issue arises at a prior stage: when menus are endogenously assigned, the analyst does not observe the counterfactual stochastic choice rule $\rho^*$ itself, but only the selected rule $\rho$. The identification of the distribution over preferences is then left to the existing literature above. My contribution here is to show what the data reveal about $\rho^*$ and the menu assignment rule before that preference identification question can even be posed.

\section{Framework}\label{sec:frame}
Let $X$ be a finite set of alternatives with $|X| = n \ge 2$, and let $\cD \subseteq 2^X \setminus \{\emptyset\}$ denote the collection of menus observed in the data. Let $\cL$ be the set of strict linear orders on $X$; a generic order is written $\succ$, and $\max(A, \succ)$ denotes the $\succ$-best element of a menu $A$. For a finite set $S$, $\Delta(S)$ denotes the set of probability distributions on $S$.

Behavior follows RUM: each DM is characterized by a preference type $\succ \in \cL$ and, when facing menu $A$, chooses $\max(A, \succ)$. For $A \subseteq X$ and $a \in A$, define the \emph{top cone}
\[
\cT(a, A)
=
\{\succ \in \cL : a \succ b \text{ for all } b \in A \setminus \{a\}\},
\]
the set of preference orders under which $a$ is chosen from $A$; equivalently, $\cT(a, A) = \{\succ \in \cL : \max(A, \succ) = a\}$.

The key primitive of my framework is the \textit{menu-assignment process}, which specifies who (preference type) faces what (menu): a joint distribution $\pi \in \Delta(\cL \times \cD)$ over (type, menu) pairs. Its marginals are the population preference distribution
\[
\nu(\succ)
=
\sum_{A \in \cD} \pi(\succ, A),
\qquad \succ \in \cL,
\]
and the menu frequency
\[
\mu(A)
=
\sum_{\succ \in \cL} \pi(\succ, A),
\qquad A \in \cD.
\]
Since all the menus in $\cD$ are observed, we have $\mu(A) > 0$ for all $A \in \cD$, so that the conditional type distribution $\pi(\succ \mid A) = \pi(\succ, A) / \mu(A)$ is well defined.

The analyst observes the pair $(\mu, \rho)$, where
\begin{equation}\label{eq:data}
\rho(a \mid A)
=
\pi\big(\cT(a, A) \mid A\big)
=
\sum_{\succ \in \cT(a, A)} \pi(\succ \mid A)
\end{equation}
is the conditional choice probability. Equivalently, the data reveal the joint distribution of menus and choices,
\[
d(A, a)
=
\mu(A)\, \rho(a \mid A)
=
\sum_{\succ \in \cT(a, A)} \pi(\succ, A).
\]

As argued in Section~\ref{sec:intro}, when menus are endogenously assigned, the observed conditional choice probability $\rho(a \mid A)$ is not the causal object of interest. The causal estimand is instead the \emph{counterfactual choice probability}
\[
\rs(a \mid A)
=
\nu\big(\cT(a, A)\big),
\]
the share of the population that would choose $a$ from menu $A$ if $A$ were assigned independently of preferences.\footnote{The discrepancy $\rho(a \mid A) - \rs(a \mid A)$ is precisely the selection bias induced by endogenous menu assignment.} Note that $\rs(\cdot \mid A)$ is well defined for every nonempty $A \subseteq X$, even for menus outside of $\cD$.

Under this framework, the implicit assumption in the stochastic choice literature corresponds to the special case in which menu assignment is independent of preferences:
\begin{equation}\label{eq:indep}
\pi = \nu \otimes \mu,
\qquad \text{that is,} \qquad
\pi(\succ, A) = \nu(\succ)\, \mu(A)
\quad \text{for every } \succ \in \cL,\ A \in \cD.
\end{equation}
Under \eqref{eq:indep}, the conditional type distribution facing any menu coincides with the population distribution, $\pi(\cdot \mid A) = \nu$, and hence
\[
\rho(a \mid A)
=
\pi\big(\cT(a, A) \mid A\big)
=
\frac{\nu\big(\cT(a, A)\big)\, \mu(A)}{\mu(A)}
=
\nu\big(\cT(a, A)\big)
=
\rs(a \mid A).
\]
Thus the standard interpretation of stochastic choice data implicitly treats menu assignment as a randomized design: selection bias vanishes, and the observed choice rule $\rho$ itself identifies the counterfactual choice rule $\rs$.

\section{Identification}\label{sec:id}

Having described the data-generating process, I now turn to the inverse problem: what can the analyst learn about the population preference distribution $\nu$, and hence about the counterfactual choice probabilities $\rs$, from the observables $(\mu, \rho)$?

Two notational conventions are used throughout. For any set of orders
$\cR\subseteq\cL$, write $\nu\left(\cR\right)=\sum_{\succ\in\cR}\nu\left(\succ\right)$,
and similarly write $\pi\left(\succ\mid A\right)$ for
$\pi\left(\left\{\succ\right\}\mid A\right)$. Because maximization is
maintained throughout, I refer to $\nu$ itself as the RUM when no
confusion arises.

\subsection{Any Data Can Be Rationalized}
Equation \eqref{eq:data} describes the data-generating process: a joint distribution $\pi$ produces the observables $(\mu, \rho)$. The analyst, however, observes $(\mu, \rho)$ only, not $\nu$, which is used for constructing $\rs$. The inverse problem is to determine which structural joint distributions could have produced the observed data. Following the language of revealed preference analysis, I call this inverse relation \emph{rationalization}. 

\begin{definition}[Rationalization and consistent RUMs]\label{def:rat}
A joint distribution $\pi \in \Delta(\cL \times \cD)$ \emph{rationalizes} $(\mu, \rho)$ if
\[
\sum_{\succ \in \cL} \pi(\succ, A) = \mu(A)
\qquad \text{and} \qquad
\pi\big(\cT(a, A) \mid A\big) = \rho(a \mid A)
\]
for every $A \in \cD$ and every $a \in A$. For a rationalizing $\pi$, let $\nu_\pi$ denote its preference marginal, $\nu_\pi(\succ) = \sum_{A \in \cD} \pi(\succ, A)$. The \emph{set of consistent RUMs} is
\[
\Theta(\mu, \rho)
=
\big\{\nu_\pi \in \Delta(\cL) : \pi \text{ rationalizes } (\mu, \rho)\big\}.
\]
\end{definition}

The first task is to establish the severity of the identification problem when menus are endogenously assigned. The classical stochastic choice literature treats axioms such as Regularity and the Block--Marschak inequalities as testable restrictions for RUM on the data. Those axioms, however, are derived under the implicit assumption that menus are exogenously assigned. Theorem~\ref{thm:anything} asks what empirical content the axioms retain when that assumption is dropped. The answer is: none. It is not merely that some RUM-inconsistent data patterns can be rationalized, but that \emph{every} data pattern can be rationalized by a population of rational DMs, no matter how anomalous it appears.

\begin{theorem}[Universal rationalizability]\label{thm:anything}
For any dataset $(\mu, \rho)$ with $\mu \in \Delta(\cD)$ and $\rho(\cdot \mid A) \in \Delta(A)$ for each $A \in \cD$, there exists a rationalizing $\pi$. That is, $\Theta(\mu, \rho) \neq \emptyset$.
\end{theorem}

\begin{proof}
For each pair $(A, a)$ with $A \in \cD$ and $a \in A$, fix an order $\succ^{(A,a)}$ that ranks $a$ above every other alternative in $X$, so that $\succ^{(A,a)} \in \cT(a, A)$. Define $\pi$ by placing mass $d(A, a) = \mu(A)\, \rho(a \mid A)$ on the pair $\big(\succ^{(A,a)}, A\big)$ for every such cell $(A, a)$. Since $\sum_{A \in \cD} \sum_{a \in A} d(A, a) = \sum_{A \in \cD} \mu(A) = 1$, the measure $\pi$ is a well-defined element of $\Delta(\cL \times \cD)$.

First, for every $A \in \cD$,
\[
\sum_{\succ \in \cL} \pi(\succ, A)
=
\sum_{a \in A} \mu(A)\, \rho(a \mid A)
=
\mu(A).
\]
Second, fix $A \in \cD$ and $a \in A$. Because every order has a unique $\succ$-best element in $A$, the top cones $\{\cT(b, A) : b \in A\}$ partition $\cL$; in particular, $\succ^{(A,b)} \in \cT(b, A)$ lies outside $\cT(a, A)$ for every $b \neq a$. Hence the only mass assigned to menu $A$ inside $\cT(a, A)$ is $d(A, a)$, and
\[
\pi\big(\cT(a, A) \mid A\big)
=
\frac{d(A, a)}{\mu(A)}
=
\rho(a \mid A).
\]
Thus $\pi$ rationalizes $(\mu, \rho)$.
\end{proof}

Theorem~\ref{thm:anything} is an ``anything goes'' result, which indicates that rationality is never falsified: since $\Theta(\mu, \rho)$ is nonempty for every dataset, the axioms of stochastic choice lose all testable content once the exogeneity of menu assignment is dropped. A violation of any such axiom may reflect nothing more than selection into menus, exactly as in the motivating Example~\ref{sec:example}. Perhaps more constructively, however, the theorem does not say the data are uninformative at all. Although rationality cannot be tested, the observables $(\mu, \rho)$ still restrict the set of consistent RUMs $\Theta(\mu, \rho)$, and through it the counterfactual choice probabilities $\rs$.

\subsection{Sharp Bounds on the Estimand and Menu Assignment Rule}

Having established what cannot be learned from observational data, I now turn to what can. The observed choice frequencies on each menu constrain which preference types could have been present there, and these constraints carve out a subset of $\Delta(\cL)$: the identified set $\Theta(\mu, \rho)$. Formally, let
\[
\cG = \{(A, a) : A \in \cD,\ a \in A\}
\]
denote the set of observed menu--choice \emph{cells}, and recall that cell $(A, a)$ carries probability mass $d(A, a) = \mu(A)\, \rho(a \mid A)$. Under maximization, the mass in cell $(A, a)$ can be supplied only by types in the top cone $\cT(a, A)$. A candidate preference distribution $\nu$ is therefore consistent with the data if and only if the masses $\nu(\succ)$, $\succ \in \cL$, can be transported to the cell demands $d(A, a)$, $(A, a) \in \cG$, with the mass of type $\succ$ permitted to flow into cell $(A, a)$ only if $\succ \in \cT(a, A)$. This is a transportation problem, and its feasibility conditions yield the identified set.

\begin{theorem}[Sharp identified set]\label{thm:sharp}
Let $\nu \in \Delta(\cL)$. The following are equivalent:
\begin{enumerate}
\item[(i)] $\nu \in \Theta(\mu, \rho)$;
\item[(ii)] there exists a map $f : \cL \times \cG \to \mathbb{R}_{+}$ with $f(\succ, (A, a)) > 0$ only if $\succ \in \cT(a, A)$, such that
\begin{align*}
\sum_{(A, a) \in \cG} f\big(\succ, (A, a)\big)
&= \nu(\succ)
\text{ for all } \succ \in \cL,\\
\sum_{\succ \in \cL} f\big(\succ, (A, a)\big)
&= d(A, a)
\text{ for all } (A, a) \in \cG.
\end{align*}

\item[(iii)] for every $\cE \subseteq \cG$,
\begin{equation}\label{eq:hall}
\sum_{(A, a) \in \cE} \mu(A)\, \rho(a \mid A)
\;\le\;
\nu\Big( \bigcup_{(A, a) \in \cE} \cT(a, A) \Big).
\end{equation}
\end{enumerate}
\end{theorem}
\vspace{0.5cm}
The proof of Theorem~\ref{thm:sharp} uses a classical feasibility result for transportation problems, the Feasibility Theorem of Gale (1957), which I state directly in the present setting: the supplies are the type masses $\nu\left(\succ\right)$, the demands are the cell masses $d\left(A,a\right)$, and mass may flow from a type to a cell only if the type belongs to the cell's top cone.\footnote{Gale's (1957) theorem: let $S$ and $T$ be finite sets, let $s:S\to\mathbb{R}_{+}$ and $t:T\to\mathbb{R}_{+}$ satisfy $\sum_{\sigma\in S}s\left(\sigma\right)=\sum_{\tau\in T}t\left(\tau\right)$, and let $E\subseteq S\times T$ be a set of admissible pairs. Then there exists $f:S\times T\to\mathbb{R}_{+}$, positive only on $E$, with row sums $s$ and column sums $t$, if and only if $\sum_{\tau\in U}t\left(\tau\right)\le\sum_{\sigma\in N\left(U\right)}s\left(\sigma\right)$ for every $U\subseteq T$, where $N\left(U\right)=\left\{\sigma\in S:\left(\sigma,\tau\right)\in E\text{ for some }\tau\in U\right\}$. Lemma~\ref{lem:gale} is the case $S=\cL$, $T=\cG$, $s=\nu$, $t=d$, and $E=\left\{\left(\succ,\left(A,a\right)\right):\succ\in\cT\left(a,A\right)\right\}$: the balance condition holds because $\sum_{\succ\in\cL}\nu\left(\succ\right)=1=\sum_{A\in\cD}\mu\left(A\right)\sum_{a\in A}\rho\left(a\mid A\right)=\sum_{\left(A,a\right)\in\cG}d\left(A,a\right)$, and the neighborhood of a collection $\cE\subseteq\cG$ is $N\left(\cE\right)=\bigcup_{\left(A,a\right)\in\cE}\cT\left(a,A\right)$.}

\begin{lemma}[Gale, 1957]\label{lem:gale}
Let $\nu\in\Delta\left(\cL\right)$. There exists $f:\cL\times\cG\to\mathbb{R}_{+}$ with $f\left(\succ,\left(A,a\right)\right)>0$ only if $\succ\in\cT\left(a,A\right)$, such that
\begin{align*}
\sum_{\left(A,a\right)\in\cG}f\left(\succ,\left(A,a\right)\right)
&=\nu\left(\succ\right)
&&\text{for every }\succ\in\cL,\\
\sum_{\succ\in\cL}f\left(\succ,\left(A,a\right)\right)
&=d\left(A,a\right)
&&\text{for every }\left(A,a\right)\in\cG,
\end{align*}
if and only if inequality \eqref{eq:hall} holds for every $\cE\subseteq\cG$.
\end{lemma}

\begin{proof}[Proof of Theorem~\ref{thm:sharp}]

With Lemma \ref{lem:gale}, the proof of the theorem consists of two steps: first, rationalizing joint distributions with marginal $\nu$ are in one-to-one correspondence with the arrays $f$ of statement (ii); second, Lemma \ref{lem:gale} converts the existence of such an array into the inequalities \eqref{eq:hall}.

\emph{(i) $\Leftrightarrow$ (ii).} The array \(f\) can be viewed as a matrix whose rows are preference types
\(\succ\in\cL\) and whose columns are observed menu--choice cells
\((A,a)\in\cG\). The entry \(f(\succ,(A,a))\) is the mass of type
\(\succ\) assigned to menu \(A\) and observed choosing \(a\). This entry can be positive only if
\(a\) is the \(\succ\)-best element of \(A\). To present it clearly, let \(\cD=\{A_1,\dots,A_J\}\), for each observed menu \(A_j\), write
\[
A_j=\{a_{j1},\dots,a_{jK_j}\}.
\]
Thus the observed menu--choice cells are
\[
\cG
=
\{(A_j,a_{jk}) : j=1,\dots,J,\ k=1,\dots,K_j\}.
\]

The transport plan \(f\) can be represented schematically as follows:
\[
\begin{array}{c|cccc|c}
&
(A_1,a_{11})
&
\cdots
&
(A_j,a_{jk})
&
\cdots
&
\text{}
\\
\hline
\succ_1
&
f(\succ_1,(A_1,a_{11}))
&
\cdots
&
f(\succ_1,(A_j,a_{jk}))
&
\cdots
&
\nu(\succ_1)
\\
\succ_2
&
f(\succ_2,(A_1,a_{11}))
&
\cdots
&
f(\succ_2,(A_j,a_{jk}))
&
\cdots
&
\nu(\succ_2)
\\
\vdots
&
\vdots
&
\ddots
&
\vdots
&
\ddots
&
\vdots
\\
\succ_{|\cL|}
&
f(\succ_{|\cL|},(A_1,a_{11}))
&
\cdots
&
f(\succ_{|\cL|},(A_j,a_{jk}))
&
\cdots
&
\nu(\succ_{|\cL|})
\\
\hline
\text{}
&
d(A_1,a_{11})
&
\cdots
&
d(A_j,a_{jk})
&
\cdots
&
1
\end{array}
\]

The entry \(f(\succ_i,(A_j,a_{jk}))\) is the mass of type \(\succ_i\)
assigned to menu \(A_j\) and observed choosing \(a_{jk}\). The row sums
require that all mass of each preference type be allocated across observed
cells:
\[
\sum_{j=1}^{J}\sum_{k=1}^{K_j}
f(\succ_i,(A_j,a_{jk}))
=
\nu(\succ_i)
\qquad
\text{for every } i=1,\dots,|\cL|.
\]
The column sums require that the total mass allocated to each observed cell
equal its observed probability:
\[
\sum_{i=1}^{|\cL|}
f(\succ_i,(A_j,a_{jk}))
=
d(A_j,a_{jk})
=
\mu(A_j)\rho(a_{jk}\mid A_j)
\qquad
\text{for every } j=1,\dots,J,\ k=1,\dots,K_j.
\]
Finally, the support restriction imposes maximization:
\[
f(\succ_i,(A_j,a_{jk}))=0
\qquad
\text{whenever }
\succ_i\notin \cT(a_{jk},A_j).
\]
Thus type \(\succ_i\) can contribute positive mass to cell
\((A_j,a_{jk})\) only if \(a_{jk}\) is the \(\succ_i\)-maximal element of
\(A_j\).

For $(ii)\Rightarrow(i)$, given $f$ as in (ii), define $\pi(\succ, A) = \sum_{a \in A} f(\succ, (A, a))$. Three checks establish that $\pi$ rationalizes $(\mu, \rho)$ with preference marginal $\nu$. First,
\[
\sum_{A \in \cD} \pi(\succ, A)
=
\sum_{(A, a) \in \cG} f\big(\succ, (A, a)\big)
=
\nu(\succ).
\]
Second,
\[
\sum_{\succ \in \cL} \pi(\succ, A)
=
\sum_{a \in A} \sum_{\succ \in \cL} f\big(\succ, (A, a)\big)
=
\sum_{a \in A} d(A, a)
=
\mu(A).
\]
Third, fix $a \in A$ and take any $\succ \in \cT(a, A)$. The support restriction implies $f(\succ, (A, b)) = 0$ for every $b \in A \setminus \{a\}$, since $\succ \notin \cT(b, A)$. Hence $\pi(\succ, A) = f(\succ, (A, a))$ for such $\succ$, and
\[
\mu(A)\, \pi\big(\cT(a, A) \mid A\big)
=
\sum_{\succ \in \cT(a, A)} \pi(\succ, A)
=
\sum_{\succ \in \cT(a, A)} f\big(\succ, (A, a)\big)
=
\sum_{\succ \in \cL} f\big(\succ, (A, a)\big)
=
d(A, a),
\]
where the third equality again uses the support restriction. Dividing by $\mu(A) > 0$ gives $\pi(\cT(a, A) \mid A) = \rho(a \mid A)$.

For $(i)\Rightarrow(ii)$, given a rationalizing $\pi$ with preference marginal $\nu$, define
\[
f\big(\succ, (A, a)\big)
=
\pi(\succ, A)\, \mathbf{1}\{\succ \in \cT(a, A)\}.
\]
The support restriction holds by construction. The column sums equal $\mu(A)\, \pi(\cT(a, A) \mid A) = d(A, a)$. For the row sums, note that for each fixed $A$ the top cones $\{\cT(a, A) : a \in A\}$ partition $\cL$, so each $\succ$ contributes $\pi(\succ, A)$ to exactly one cell $(A, a)$; summing over $A \in \cD$ gives $\nu(\succ)$.

\emph{(ii) $\Leftrightarrow$ (iii).} This equivalence is Lemma~\ref{lem:gale}.
\end{proof}

The first implication of Theorem~\ref{thm:sharp} is the identification of sharp bounds of $\rs$. The values of $\rs\left(b\mid B\right)$ that are consistent with the data are exactly the values $\nu\left(\cT\left(b,B\right)\right)$ generated by some $\nu\in\Theta\left(\mu,\rho\right)$; the next corollary shows that this set of values is a closed interval and identifies its endpoints.

\begin{corollary}\label{cor:bounds}
For any nonempty $B\subseteq X$ and any $b\in B$, the identified set for $\rs\left(b\mid B\right)$ is the closed interval
\[
\rs\left(b\mid B\right)\in
\left[
\min_{\nu\in\Theta\left(\mu,\rho\right)}
\nu\left(\cT\left(b,B\right)\right),
\quad
\max_{\nu\in\Theta\left(\mu,\rho\right)}
\nu\left(\cT\left(b,B\right)\right)
\right].
\]
Both endpoints are attained and are values of finite-dimensional linear programs.
\end{corollary}

\begin{proof}
For any \(\nu\in\Theta(\mu,\rho)\), the counterfactual choice probability is
\[
\rho^*(b\mid B)=\nu(\cT(b,B)).
\]
Hence the identified set for \(\rho^*(b\mid B)\) is
\[
\left\{\nu(\cT(b,B)):\nu\in\Theta(\mu,\rho)\right\}.
\]
By Theorem~\ref{thm:sharp}(iii), $\Theta(\mu,\rho)$ is defined by finitely
many weak linear inequalities within $\Delta(\cL)$, hence it is a compact
convex polytope; it is nonempty by Theorem~\ref{thm:anything}. Because
\(\Theta(\mu,\rho)\) is compact and convex, and because \(\nu(\cT(b,B))\)
is linear in \(\nu\), this set is a closed interval. Its endpoints are
obtained by minimizing and maximizing \(\nu(\cT(b,B))\) over all
\(\nu\in\Theta(\mu,\rho)\). The endpoints are attained because
$\Theta(\mu,\rho)$ is compact, and they can be computed as
finite-dimensional linear programs using the flow formulation in
Theorem~\ref{thm:sharp}(ii).
\end{proof}

The interval in Corollary~\ref{cor:bounds} contains every value of $\rs\left(b\mid B\right)$ that some rationalizing joint distribution can generate, so the data and rationality alone can never justify a narrower interval. An attractive and interesting feature of Corollary~\ref{cor:bounds} is that it does not require the observed cell mass $d(B,b)=\mu(B)\rho(b\mid B)$ to be positive. If $\mu(B)=0$ (i.e., $B\notin\cD$), no DM is observed facing menu $B$ and $\rho(b\mid B)$ is not available. The absence of menu $B$ from the data does not imply that $\rho^{*}(b\mid B)=0$: no one was assigned $B$, but some members of the population may still be types who would choose $b$ if $B$ were assigned. Such types may be partially revealed by their choices from other observed menus. The bounds then answer an extrapolation question in the spirit of Manski (2007), but from selection-contaminated data. If $\rho(b\mid B)=0$, no DM who is assigned menu $B$ chooses $b$. But under endogenous menu assignment, this need not mean that no one in the population would choose $b$ from $B$: such preference types may exist, but may be selected into other menus. Thus $\rho^{*}(b\mid B)$ can be positive even when $\rho(b\mid B)=0$, and the corollary still gives the sharp set of such counterfactual values.

The second implication of Theorem~\ref{thm:sharp} concerns the assignment mechanism itself: the same optimization that bounds $\rs$ also bounds the degree of endogeneity in the data. For any rationalizing $\pi$ with preference marginal $\nu$, define the \emph{steering probability}
\[
\pi\left(A\mid\cT\left(a,A\right)\right)
=
\frac{
\sum_{\succ\in\cT\left(a,A\right)}\pi\left(\succ,A\right)
}{
\nu\left(\cT\left(a,A\right)\right)
},
\]
the probability that a DM whose preference type ranks $a$ at the top of $A$ is assigned menu $A$. The steering probability has a natural benchmark: under independent assignment, $\pi=\nu\otimes\mu$, it equals $\mu\left(A\right)$ for every $\left(A,a\right)$. Values above $\mu\left(A\right)$ therefore reveal that the types who favor $a$ within $A$ are over-represented among those shown $A$; in the motivating example, this is precisely how the platform's algorithm steers sushi lovers toward the augmented feed.

The steering probability depends on $\pi$ only through $\nu\left(\cT\left(a,A\right)\right)$, and is therefore bounded by the same program as in Corollary~\ref{cor:bounds}.\footnote{For steering, unlike for $\rho^{*}(a\mid A)$ itself, I require $d(A,a)>0$. This guarantees that the conditioning event $\cT(a,A)$ has positive population mass in every rationalization, since $\rho^{*}(a\mid A)=\nu(\cT(a,A))\ge d(A,a)>0$. If $d(A,a)=0$, the data may be consistent either with no such types in the population, in which case $\pi(A\mid\cT(a,A))$ is undefined, or with such types existing but never being assigned menu $A$, in which case the steering probability is zero. Thus the ratio formula is well defined uniformly over the identified set only for positive observed cells.}

\begin{corollary}\label{cor:steering-bounds}
Fix $A\in\cD$ and $a\in A$ with $d\left(A,a\right)=\mu\left(A\right)\rho\left(a\mid A\right)>0$. Then every rationalizing $\pi$ with preference marginal $\nu$ satisfies
\[
\pi\left(A\mid\cT\left(a,A\right)\right)
=
\frac{d\left(A,a\right)}{\nu\left(\cT\left(a,A\right)\right)},
\]
and the identified set for $\pi\left(A\mid\cT\left(a,A\right)\right)$ is the closed interval
\[
\left[
\frac{
d\left(A,a\right)
}{
\max_{\nu\in\Theta\left(\mu,\rho\right)}
\nu\left(\cT\left(a,A\right)\right)
},
\quad
\frac{
d\left(A,a\right)
}{
\min_{\nu\in\Theta\left(\mu,\rho\right)}
\nu\left(\cT\left(a,A\right)\right)
}
\right].
\]
Both endpoints are attained.
\end{corollary}

\begin{proof}
Fix a rationalizing $\pi$ with preference marginal $\nu$. By the
definition of the steering probability, its numerator is
$\sum_{\succ\in\cT\left(a,A\right)}\pi\left(\succ,A\right)
=\mu\left(A\right)\,\pi\big(\cT\left(a,A\right)\mid A\big)
=\mu\left(A\right)\rho\left(a\mid A\right)=d\left(A,a\right)$, where the
second equality holds because $\pi$ rationalizes $\left(\mu,\rho\right)$.
Moreover, the denominator is positive: applying \eqref{eq:hall} to the
single cell $\cE=\left\{\left(A,a\right)\right\}$ gives
$\nu\left(\cT\left(a,A\right)\right)\ge d\left(A,a\right)>0$. This proves
the displayed identity and shows that the steering probability depends on
$\pi$ only through its preference marginal $\nu$.
Consequently, the identified set for
$\pi\left(A\mid\cT\left(a,A\right)\right)$ is
$\left\{d\left(A,a\right)/\nu\left(\cT\left(a,A\right)\right)
:\nu\in\Theta\left(\mu,\rho\right)\right\}$. By
Corollary~\ref{cor:bounds}, as $\nu$ ranges over
$\Theta\left(\mu,\rho\right)$, the denominator
$\nu\left(\cT\left(a,A\right)\right)$ takes every value in the closed
interval between its minimum and its maximum. Dividing
$d\left(A,a\right)$ by $\nu\left(\cT\left(a,A\right)\right)$ gives the stated interval, with both endpoints attained by the
distributions that attain the endpoints in Corollary~\ref{cor:bounds}.
\end{proof}

Corollary~\ref{cor:steering-bounds} turns the bounds of Corollary~\ref{cor:bounds} into a diagnostic for menu endogeneity. Dividing by the benchmark $\mu\left(A\right)$ gives
\[
\frac{\pi\left(A\mid\cT\left(a,A\right)\right)}{\mu\left(A\right)}
=
\frac{\rho\left(a\mid A\right)}{\nu\left(\cT\left(a,A\right)\right)}
=
\frac{\rho\left(a\mid A\right)}{\rs\left(a\mid A\right)},
\]
so the steering probability exceeds its independence benchmark exactly when the observed choice probability exceeds the counterfactual one. If the lower endpoint of the identified interval exceeds $\mu\left(A\right)$, the data alone establish that menu $A$ was disproportionately assigned to the types who favor $a$ within it, without any assumption on the assignment rule.

\subsection{Sharp Bounds in Closed Form}\label{sec:closed-form}

Corollary~\ref{cor:bounds} characterizes the sharp bounds as the values of
linear programs, but these programs are enormous. The unknown is the vector
of type probabilities $\left(\nu\left(\succ\right)\right)_{\succ\in\cL}$,
which has $n!$ entries, and by Theorem~\ref{thm:sharp} the feasible set is
described by one linear inequality for every collection of observed cells
$\cE\subseteq\cG$, of which there are $2^{\left|\cG\right|}$. In this
subsection I show that the bounds have a closed-form expression.

The key object is the set function that aggregates the evidence the data provide about any given set of preference types. For $\cR\subseteq\cL$, define
\[
r\left(\cR\right)
=
\sum_{\left(A,a\right)\in\cG:\ \cT\left(a,A\right)\subseteq\cR}
d\left(A,a\right),
\]
with $r\left(\emptyset\right)=0$ and $r\left(\cL\right)=1$.\footnote{Each top cone is nonempty, so no cell contributes to $r\left(\emptyset\right)$; every cell contributes to $r\left(\cL\right)$, whose value is $\sum_{A\in\cD}\mu\left(A\right)\sum_{a\in A}\rho\left(a\mid A\right)=1$.}

\begin{lemma}[Supermodularity]\label{lem:super}
The function $r$ is supermodular: for all $\cR,\cR'\subseteq\cL$,
\[
r\left(\cR\cup\cR'\right)
+
r\left(\cR\cap\cR'\right)
\ge
r\left(\cR\right)
+
r\left(\cR'\right).
\]
\end{lemma}

\begin{proof}
For every $\cR\subseteq\cL$, the definition of $r$ can be rewritten with the summation running over all cells,
\[
r\left(\cR\right)
=
\sum_{\left(A,a\right)\in\cG}
d\left(A,a\right)\,
\1\left\{\cT\left(a,A\right)\subseteq\cR\right\},
\]
since a cell whose cone is not contained in $\cR$ contributes zero. Hence, for any $\cR,\cR'\subseteq\cL$,
\[
r\left(\cR\cup\cR'\right)+r\left(\cR\cap\cR'\right)
-r\left(\cR\right)-r\left(\cR'\right)
=
\sum_{\left(A,a\right)\in\cG}
d\left(A,a\right)\,
\delta\left(A,a\right),
\]
where, for each cell $\left(A,a\right)\in\cG$,
\begin{align*}
\delta\left(A,a\right)
={}&
\1\left\{\cT\left(a,A\right)\subseteq\cR\cup\cR'\right\}
+\1\left\{\cT\left(a,A\right)\subseteq\cR\cap\cR'\right\}\\
&-\1\left\{\cT\left(a,A\right)\subseteq\cR\right\}
-\1\left\{\cT\left(a,A\right)\subseteq\cR'\right\}.
\end{align*}
Since $d\left(A,a\right)\ge 0$, it suffices to show that $\delta\left(A,a\right)\ge 0$ for every cell. If $\cT\left(a,A\right)$ is contained in both $\cR$ and $\cR'$, then it is contained in $\cR\cap\cR'$ and in $\cR\cup\cR'$, so $\delta\left(A,a\right)=1+1-1-1=0$. If $\cT\left(a,A\right)$ is contained in exactly one of $\cR$ and $\cR'$, then it is contained in $\cR\cup\cR'$, so $\delta\left(A,a\right)\ge 1+0-1-0=0$. If $\cT\left(a,A\right)$ is contained in neither, both subtracted indicators are zero, so $\delta\left(A,a\right)\ge 0$.
\end{proof}

The interpretation of $r$ is as follows. Under maximization, the mass $d\left(A,a\right)$ of cell $\left(A,a\right)$ can only be supplied by types in the top cone $\cT\left(a,A\right)$: the data locate this mass inside the cone, but are silent about which types within the cone supply it. The value $r\left(\cR\right)$ then collects the mass of every cell whose cone lies entirely inside $\cR$. Any preference distribution consistent with the data must give $\cR$ at least this much mass, $\nu\left(\cR\right)\ge r\left(\cR\right)$, and the two quantities relevant for the bounds are the values of $r$ at a top cone and at its complement. With $r$ being supermodular, I shall use a corollary of the theorem due to Shapley (1971), adapted for my
setting as follows.\footnote{Shapley's (1971) theorem states that for a supermodular $w:2^{N}\to\mathbb{R}$ with $w\left(\emptyset\right)=0$ on a finite set $N$, the set of vectors $x\in\mathbb{R}^{N}$ with $\sum_{i\in N}x_{i}=w\left(N\right)$ and $\sum_{i\in S}x_{i}\ge w\left(S\right)$ for all $S\subseteq N$ is nonempty, and $\min\sum_{i\in S}x_{i}=w\left(S\right)$ and $\max\sum_{i\in S}x_{i}=w\left(N\right)-w\left(N\setminus S\right)$ over this set, both attained. Lemma~\ref{lem:shapley} is the case $N=\cL$, $w=r$: since $r\left(\left\{\succ\right\}\right)\ge 0$ for every $\succ$ and $r\left(\cL\right)=1$, the vectors in question are exactly the probability distributions dominating $r$.}

\begin{lemma}[Shapley, 1971]\label{lem:shapley}
For a supermodular $r$ with $r\left(\emptyset\right)=0$ and $r\left(\cL\right)=1$, the set
\[
\operatorname{core}\left(r\right)
=
\left\{
\nu\in\Delta\left(\cL\right)
:
\nu\left(\cR\right)\ge r\left(\cR\right)
\ \text{for every }\cR\subseteq\cL
\right\}
\]
is nonempty, and for every $\cR\subseteq\cL$,
\[
\min_{\nu\in\operatorname{core}\left(r\right)}\nu\left(\cR\right)
=
r\left(\cR\right),
\qquad
\max_{\nu\in\operatorname{core}\left(r\right)}\nu\left(\cR\right)
=
1-r\left(\cL\setminus\cR\right),
\]
and both values are attained.
\end{lemma}

Lemma~\ref{lem:shapley} gives the range of $\nu\left(\cR\right)$ over $\operatorname{core}\left(r\right)$. The main result of this subsection provides the closed forms of the bounds in Corollary \ref{cor:bounds}.

\begin{theorem}[Closed-form bounds]\label{thm:closedform}
Fix $A\subseteq X$ with $\left|A\right|\ge 2$ and $a\in A$. Then:
\begin{enumerate}
\item[(i)] the identified set coincides with the core:
$\Theta\left(\mu,\rho\right)=\operatorname{core}\left(r\right)$;
\item[(ii)] the sharp bounds on the counterfactual choice probability are
\[
\rs\left(a\mid A\right)
\in
\Big[\,
r\left(\cT\left(a,A\right)\right),
\;
1-r\left(\cL\setminus\cT\left(a,A\right)\right)
\Big];
\]
\item[(iii)] if moreover $A\in\cD$ and $d\left(A,a\right)>0$, the sharp bounds on the steering probability are
\begin{align*}
    \pi\left(A\mid\cT\left(a,A\right)\right)
\in
\left[
\frac{d\left(A,a\right)}{1-r\left(\cL\setminus\cT\left(a,A\right)\right)},
\;
\frac{d\left(A,a\right)}{r\left(\cT\left(a,A\right)\right)}
\right],
\end{align*}

\end{enumerate}
where
\begin{align*}
 r\left(\cT\left(a,A\right)\right) 
&=\sum_{B\in\cD:\,B\supseteq A}d\left(B,a\right),\\
r\left(\cL\setminus\cT\left(a,A\right)\right) 
&=\sum_{B\in\cD:\,a\in B}\ \sum_{b\in\left(A\cap B\right)\setminus\left\{a\right\}}d\left(B,b\right).
\end{align*}
\end{theorem}

\begin{proof}
\emph{For (i).}
Let $\nu\in\Theta\left(\mu,\rho\right)$ and fix $\cR\subseteq\cL$. The cells whose cones lie in $\cR$ form the collection $\cE_{\cR}=\left\{\left(B,b\right)\in\cG:\cT\left(b,B\right)\subseteq\cR\right\}$, whose total mass is $r\left(\cR\right)$ by definition, and whose cones have union contained in $\cR$. Inequality \eqref{eq:hall} of Theorem~\ref{thm:sharp}, applied to $\cE_{\cR}$, then gives $r\left(\cR\right)\le\nu\left(\cR\right)$, so $\nu\in\operatorname{core}\left(r\right)$. Conversely, let $\nu\in\operatorname{core}\left(r\right)$ and fix any $\cE\subseteq\cG$. Its cones have union $\cR_{\cE}=\bigcup_{\left(B,b\right)\in\cE}\cT\left(b,B\right)$, and every cell of $\cE$ has its cone contained in $\cR_{\cE}$, so $\cE\subseteq\cE_{\cR_{\cE}}$ and hence $\sum_{\left(B,b\right)\in\cE}d\left(B,b\right)\le r\left(\cR_{\cE}\right)\le\nu\left(\cR_{\cE}\right)$, which is \eqref{eq:hall} for $\cE$. Since $\cE$ was arbitrary, Theorem~\ref{thm:sharp} gives $\nu\in\Theta\left(\mu,\rho\right)$.

\emph{For (ii).}
By Corollary~\ref{cor:bounds}, the sharp bounds on $\rs\left(a\mid A\right)$ are the minimum and the maximum of $\nu\left(\cT\left(a,A\right)\right)$ over $\Theta\left(\mu,\rho\right)$, which by (i) and Lemma~\ref{lem:shapley} equal $r\left(\cT\left(a,A\right)\right)$ and $1-r\left(\cL\setminus\cT\left(a,A\right)\right)$, both attained. 

\emph{For (iii).} By Corollary~\ref{cor:steering-bounds}, every
rationalization satisfies
$\pi\left(A\mid\cT\left(a,A\right)\right)
=d\left(A,a\right)/\nu\left(\cT\left(a,A\right)\right)$. By (ii),
as $\nu$ ranges over $\Theta\left(\mu,\rho\right)$, the denominator
$\nu\left(\cT\left(a,A\right)\right)$ ranges from $r\left(\cT\left(a,A\right)\right)$ to $1-r\left(\cL\setminus\cT\left(a,A\right)\right)$, and this interval
excludes zero: the cell $\left(A,a\right)$ is one term of the first sum,
so $r\left(\cT\left(a,A\right)\right)\ge d\left(A,a\right)>0$. Dividing
$d\left(A,a\right)$ by $\nu\left(\cT\left(a,A\right)\right)$ gives the stated interval.

It remains to compute the two values of $r$.

For $r\left(\cT\left(a,A\right)\right)=\sum_{B\in\cD:\,B\supseteq A}d\left(B,a\right),$ I now prove that $\cT\left(b,B\right)\subseteq\cT\left(a,A\right)$ iff $b=a$ and $B\supseteq A$. For $\Leftarrow$, if $\succ$ ranks $a$ on top of $B$, it ranks $a$ on top of $A\subseteq B$. For $\Rightarrow$ part, if $b\ne a$, any order ranking $b$ first and $a$ last lies in $\cT\left(b,B\right)$ but not in $\cT\left(a,A\right)$, since $\left|A\right|\ge 2$. If $b=a$ but $B\not\supseteq A$, take $x\in A\setminus B$: any order ranking $x$ first and $a$ second lies in $\cT\left(a,B\right)$, as $x\notin B$, but not in $\cT\left(a,A\right)$, as $x\in A$. Hence the cells with cone inside $\cT\left(a,A\right)$ are exactly $\left\{\left(B,a\right):B\in\cD,\,B\supseteq A\right\}$, and summing their masses gives the first sum.

For $r\left(\cL\setminus\cT\left(a,A\right)\right)=\sum_{B\in\cD:\,a\in B}\ \sum_{b\in\left(A\cap B\right)\setminus\left\{a\right\}}d\left(B,b\right)$, note that $\cT\left(b,B\right)\subseteq\cL\setminus\cT\left(a,A\right)$ means $\cT\left(b,B\right)\cap\cT\left(a,A\right)=\emptyset$, and I prove that this holds iff $b\in A\setminus\left\{a\right\}$ and $a\in B$. For $\Leftarrow$ part, a common order would satisfy $a\succ b$, since $b\in A$, and $b\succ a$, since $a\in B$. For $\Rightarrow$ part, if $b=a$, any order ranking $a$ first lies in both cones. If $b\ne a$ and $b\notin A$, any order ranking $b$ first and $a$ second lies in both cones; if $b\ne a$ and $a\notin B$, any order ranking $a$ first and $b$ second does---in each case the first-ranked alternative is absent from the other menu, where the second-ranked one is therefore best. Hence the cells with cone disjoint from $\cT\left(a,A\right)$ are exactly $\left\{\left(B,b\right):B\in\cD,\,a\in B,\,b\in\left(A\cap B\right)\setminus\left\{a\right\}\right\}$, and summing their masses gives the second sum.
\end{proof}

\section{Point Identification under Outcome-Based Steering}\label{sec:pct}
The bounds in Section~\ref{sec:id} are the most the data can deliver
under complete agnosticism about the menu assignment/steering rule. In this section, I impose an economically motivated restriction on steering: the platform assigns menus according to the alternative it predicts the consumer will choose from each menu. I combine this restriction with cross market variation, assuming that the same assignment rule operates across several markets while the distribution of preferences varies across them. I then show that one can obtain full point
identification of the counterfactual choice
probabilities in each market and the steering rule
itself. Now I introduce the environment and notation.

There are $M$ markets, $m=1,\dots,M$, with population shares $\lambda_{m}>0$, $\sum_{m}\lambda_{m}=1$, and preference distributions $\nu_{m}\in\Delta\left(\cL\right)$. The population preference distribution is $\nu=\sum_{m}\lambda_{m}\nu_{m}$. Let $\mu_m$ denote the menu frequency in market $m$. Let $\cS=\bigcup_{m=1}^{M}\operatorname{supp}\left(\nu_{m}\right)$ denote
the set of preference types present in some market. In each market, menus are assigned by a rule $\pi_{m}\left(\cdot\mid\succ\right)\in\Delta\left(\cD\right)$, and the analyst observes the cross-sectional data of the baseline model,
\[
d_{m}\left(A,a\right)
=
\sum_{\succ\in\cT\left(a,A\right)}
\pi_{m}\left(A\mid\succ\right)\nu_{m}\left(\succ\right),
\qquad
\left(A,a\right)\in\cG,
\]
together with the pooled data $\bar{d}\left(A,a\right)=\sum_{m}\lambda_{m}d_{m}\left(A,a\right)$.
The first assumption requires the assignment rule to be common across markets.

\begin{assumption}[Common steering]\label{ass:pct-common}
The menu assignment rule is the same in every market:
$\pi_{i}=\pi_{k}$
for all $i,k\in\left\{1,\dots,M\right\}$.
\end{assumption}
Assumption~\ref{ass:pct-common} is the natural description of a single platform, broker, or chain deploying one menu assignment rule across the regions it serves: the rule maps a predicted preference profile to a menu in the same way everywhere, and markets differ only in the distribution of preference types. Under Assumption~\ref{ass:pct-common}, write $\pi$ for the common conditional assignment probability.

The second assumption states that given preference $\succ$, the assignment of menu $A$ is only determined by the choice that preference type $\succ$ will make from that menu.

\begin{assumption}[Outcome-based steering]\label{ass:pct}
For all $A\in\cD$ and for all $\succ,\succ'\in\cS$, if $\max\left(A,\succ\right)=\max\left(A,\succ'\right)$, we have $\pi\left(A\mid\succ\right)
=
\pi\left(A\mid\succ'\right)
>0.$
\end{assumption}

The interpretation is that any two DMs who would pick the same alternative from a menu are equally likely to be shown that menu, no matter how they rank the alternatives they would not pick. The platform's decision to show menu $A$ to a DM thus depends on the DM's type only through the predicted \emph{outcome} on $A$, I name this an outcome-based steering rule.

Assumption~\ref{ass:pct} is an intuitive restriction because it mirrors how menu steering rules are actually built. For example, a recommender system may score a feed by the predicted probability of a click or purchase from that feed; a salesman may choose which menu to present to their client based on which product the client will buy. In these cases the platform's knowledge about the DM enters the assignment of $A$ only
through the choice the DM is predicted to make from $A$.\footnote{It is easy to check that the motivating Example~\ref{sec:example} satisfies Assumption~\ref{ass:pct}.} Also, Assumption~\ref{ass:pct} nests exogenous menu assignment as a special case. Suppose that, for every menu
$A\in\cD$, we have $\pi\left(A\mid\succ\right)=\kappa\left(A\right)$ for all $\succ\in\cS$, 
so that the probability of being shown $A$ is independent of preference. This is a special case of Assumption~\ref{ass:pct}. Then, in any market $m$, the probability that a DM is of type $\succ$ and
is shown menu $A$ equals
\[
\pi_{m}\left(\succ,A\right)
=
\pi\left(A\mid\succ\right)\nu_{m}\left(\succ\right)
=
\kappa\left(A\right)\nu_{m}\left(\succ\right),
\]
and summing over $\succ$ gives $\mu_{m}\left(A\right)=\kappa\left(A\right)$.
Hence $\pi_{m}\left(\succ,A\right)=\nu_{m}\left(\succ\right)
\mu_{m}\left(A\right)$, which is the independence condition
\eqref{eq:indep} that Section~\ref{sec:frame} showed to underlie the
stochastic choice literature.

Under Assumption~\ref{ass:pct-common} and Assumption~\ref{ass:pct}, let $\gamma\left(A,a\right)\in\left(0,1\right]$ denote the steering rule for menu $A$ and alternative $a$, the probability that menu $A$ is assigned to a supported type whose predicted choice from $A$ is $a$. Then, $\gamma\left(A,a\right)$ is the common value of $\pi\left(A\mid\succ\right)$ across supported types $\succ\in\cT\left(a,A\right)$, so that
\[
\pi\left(A\mid\succ\right)
=
\gamma\left(A,\max\left(A,\succ\right)\right)
\text{ for all }A\in\cD\text{ and for all }\succ\in\cS.
\]

The objects of interest are the market-level counterfactual choice probabilities $\rs_{m}\left(a\mid A\right)=\nu_{m}\left(\cT\left(a,A\right)\right)$, the population counterfactual $\rs\left(a\mid A\right)=\nu\left(\cT\left(a,A\right)\right)=\sum_{m}\lambda_{m}\rs_{m}\left(a\mid A\right)$, and the steering rule $\gamma\left(A,a\right)$.\footnote{The shares $\lambda_m$ are used only for aggregation: the steering rule and the market-level counterfactuals are identified from the market-level data $(d_m)_{m=1}^M$, while $\lambda$ is required only to aggregate to the population counterfactual $\rs$.}

The key implication of Assumption~\ref{ass:pct} is a factorization of the
observed cell masses. Under the assumption, all supported types in \(\cT(a,A)\) share the same steering rule \(\gamma(A,a)\), and each observed cell frequency becomes

\begin{equation}\label{eq:collapse}
d_{m}\left(A,a\right)
=
\gamma\left(A,a\right)
\sum_{\succ\in\cT\left(a,A\right)}\nu_{m}\left(\succ\right)
=
\gamma\left(A,a\right)\,\rs_{m}\left(a\mid A\right).
\end{equation}
This factorization is the engine of the identification
argument: cross-market variation moves the second factor while holding the
first fixed, and I now show how one can separate the two.

By \eqref{eq:collapse}, $\rs_{m}\left(a\mid A\right)=d_{m}\left(A,a\right)/\gamma\left(A,a\right)$, and since the counterfactual choice probabilities on menu $A$ sum to one in every market, the steering rules must satisfy
\begin{equation}\label{eq:pct-system}
\sum_{a\in A}\rs_{m}\left(a\mid A\right)=
\sum_{a\in A}\frac{d_{m}\left(A,a\right)}{\gamma\left(A,a\right)}
=
1,
\qquad
m=1,\dots,M.
\end{equation}

Fix a menu $A$. For each market $m$, \eqref{eq:pct-system} is one linear restriction on the $\left|A\right|$ unknowns $1/\gamma\left(A,a\right)$, $a\in A$, with observed coefficients $d_{m}\left(A,a\right)$. Collecting these coefficients in the $M\times\left|A\right|$ matrix
\[
D_{A}
=
\Big[d_{m}\left(A,a\right)\Big]_{m=1,\dots,M;\ a\in A},
\]
with one row per market and one column per alternative, system \eqref{eq:pct-system} reads
\begin{equation}\label{eq:linsys}
D_{A}\,x=\1_{M},
\qquad
x=\left(1/\gamma\left(A,a\right)\right)_{a\in A},
\end{equation}
which has a unique solution whenever $\operatorname{rank}\left(D_{A}\right)=\left|A\right|$. By \eqref{eq:collapse}, the $a$-th column of $D_{A}$ is $\gamma\left(A,a\right)\cdot\big(\rs_{m}\left(a\mid A\right)\big)_{m=1}^{M}$, a nonzero multiple of the vector of cone masses across markets, so the condition holds if and only if the vectors $\big(\nu_{m}\left(\cT\left(a,A\right)\right)\big)_{m=1}^{M}$, $a\in A$, are linearly independent: markets must genuinely differ in the composition of preferences over $A$, and there must be at least as many markets as alternatives on the menu, $M\ge\left|A\right|$. Formally,

\begin{theorem}[Point identification under outcome-based steering]\label{thm:pct}
Suppose Assumptions~\ref{ass:pct-common}--\ref{ass:pct} hold. Fix $A\in\cD$ with $\operatorname{rank}\left(D_{A}\right)=\left|A\right|$. Then,  for all $a\in A$:
\begin{enumerate}
\item[(i)] the vector $x=\left(1/\gamma\left(A,a\right)\right)_{a\in A}$ is the unique solution of $D_{A}x=\1_{M}$, hence $\gamma\left(A,a\right)$ is uniquely identified;
\item[(ii)] the market-level and population counterfactual choice probabilities are point identified:
\[
\rs_{m}\left(a\mid A\right)
=
\frac{d_{m}\left(A,a\right)}{\gamma\left(A,a\right)},
\qquad
\rs\left(a\mid A\right)
=
\frac{\bar{d}\left(A,a\right)}{\gamma\left(A,a\right)}.
\]
\end{enumerate}
\end{theorem}

\begin{proof}
The true $\gamma\left(A,\cdot\right)$ solves \eqref{eq:pct-system}: by \eqref{eq:collapse},
\[
\sum_{a\in A}
\frac{d_{m}\left(A,a\right)}{\gamma\left(A,a\right)}
=
\sum_{a\in A}
\nu_{m}\left(\cT\left(a,A\right)\right)
=
1
\qquad
\text{for every }m,
\]
because the top cones $\left\{\cT\left(a,A\right):a\in A\right\}$ partition $\cL$. For uniqueness, note that if $\gamma'$ solves \eqref{eq:pct-system}, its reciprocals $1/\gamma'\left(A,a\right)$ satisfy the linear system \eqref{eq:linsys}. Since $\operatorname{rank}\left(D_{A}\right)=\left|A\right|$, that system has at most one solution, so the reciprocals, and hence $\gamma'$ itself, coincide with those of $\gamma\left(A,\cdot\right)$. This proves (i). For (ii), \eqref{eq:collapse} gives $\rs_{m}\left(a\mid A\right)=d_{m}\left(A,a\right)/\gamma\left(A,a\right)$, and averaging over markets with weights $\lambda_{m}$ gives $\rs\left(a\mid A\right)=\sum_{m}\lambda_{m}\rs_{m}\left(a\mid A\right)=\bar{d}\left(A,a\right)/\gamma\left(A,a\right)$.
\end{proof}

For any menu $A$, the rank condition on $D_{A}$ alone---whether or not it holds for any other menu---identifies $\gamma\left(A,\cdot\right)$, and with it the counterfactuals $\rs_{m}\left(\cdot\mid A\right)$ in \emph{every} market $m$, including markets that contributed no identifying variation to $D_{A}$. The Appendix applies Theorem~\ref{thm:pct} to the motivating example: with three markets, the first of which reproduces the motivating example exactly, the platform's true steering rule and the counterfactual choice probabilities of all three markets are recovered uniquely.

Unlike the unrestricted model of Theorem~\ref{thm:anything}, which can
rationalize any dataset, the targeted model is falsifiable, even when it
is exactly identified. The recovered steering rules must be probabilities,
$0<\gamma\left(A,a\right)\le 1$ for every $a\in A$;\footnote{An identified
steering rule $\gamma$ can violate this only by being negative or by exceeding one:
$\gamma\left(A,a\right)$ is the reciprocal of a coordinate of the finite
solution of \eqref{eq:linsys}, so it is never zero, and a zero column of
$D_{A}$, which, under Assumption~\ref{ass:pct}, indicates that no preference type in $\cS$ chooses $a$ from $A$, would
already fail the rank condition.} they must give each type a well-defined assignment distribution over
menus: since $\pi\left(A\mid\succ\right)
=\gamma\left(A,\max\left(A,\succ\right)\right)$, the recovered values must
satisfy
$\sum_{A\in\cD}\gamma\left(A,\max\left(A,\succ\right)\right)=1$ for every
$\succ\in\cS$; and the implied counterfactuals must
be rationalizable for each market by a single preference distribution:
for each $m$ there must exist $\nu_{m}\in\Delta\left(\cL\right)$ with
$\nu_{m}\left(\cT\left(a,A\right)\right)=\rs_{m}\left(a\mid A\right)$, which is the revealed stochastic preference problem of McFadden and
Richter (1991).\footnote{McFadden and Richter (1991) show that such a
$\nu_{m}$ exists if and only if, for every finite sequence of cells
$\left(A_{1},a_{1}\right),\dots,\left(A_{K},a_{K}\right)\in\cG$ with
repetitions allowed,
$\sum_{k=1}^{K}\rho_m^{*}(a_k\mid A_k)
\le
\max_{\succ\in\cL}
\sum_{k=1}^{K}\mathbf 1\{a_k=\max(A_k,\succ)\}.$.} None of these restrictions is guaranteed by the data, and each is derived from the
conjunction of preference maximization (the underlying RUM), common steering (Assumption~\ref{ass:pct-common}), and outcome-based steering (Assumption~\ref{ass:pct}),
so any violation falsifies the joint hypothesis without indicating which
component fails.\footnote{Since the goal is to identify the steering
rule and the counterfactuals from minimal data, $M=\left|A\right|$ markets
with $\operatorname{rank}\left(D_{A}\right)=\left|A\right|$ suffice for
menu $A$. Once the rank condition is already satisfied, additional markets do not improve point identification. Unlike settings in which extra observations are simply
redundant, here each market beyond $\left|A\right|$ adds one
overidentifying restriction---system \eqref{eq:linsys} acquires more
equations than unknowns, and its solvability becomes a falsifiable
implication of the model as well.}

\section{Concluding Remarks}\label{sec:conclusion}

For the stochastic choice literature, this paper identifies an important yet largely ignored problem and offers a solution. The problem is that any axiom is a test of a joint hypothesis of rationality \emph{and} exogenous menu assignment, so observational violations of these axioms cannot by themselves be read as evidence of behavioral phenomena. The solution requires almost nothing new from the data: the menu frequency $\mu$, recorded in many applied choice datasets and discarded at the analysis stage, is the object that converts stochastic choice data from
uninterpretable to sharply bounded, and market-level variation yields unique identification. Notwithstanding the parsimonious nature of the dataset, these observables encode substantial information. In the motivating example, the sharp bounds alone cap the seeming decoy effect at zero, and the three-market system identifies the causal estimands and the menu assignment rule uniquely in every market.

The framework thus extends the reach of revealed preference analysis from the RCT, where menus are exogenously assigned, to the scenarios in which menus are endogenously assigned. In practice, sometimes the assignment rules are dynamic: recommenders update on realized choices, so today's menu depends on yesterday's behavior. Extending the framework to a dynamic assignment process, and asking whether the platform's own learning dynamics generate the variation needed for identification, would connect this agenda to adaptive menu assignment. I leave this for future work.

\newpage
\appendix
\section{Worked Identification in the Motivating Example}\label{app:worked-identification}
This appendix carries out the paper's identification program on the motivating Example~\ref{sec:example}. Section~\ref{app:gale-identification} derives the sharp bounds on the counterfactual choice probabilities and the steering probabilities from Theorem~\ref{thm:sharp}. Section~\ref{app:closed-form-identification} shows that the closed forms of Theorem~\ref{thm:closedform} reproduce the same bounds mechanically. Section~\ref{app:pct-identification} adds the structure of Section~\ref{sec:pct} and obtains point identification from Theorem~\ref{thm:pct}.

Let $X=\{a,b,c\}$ and $\cD=\{\{a,b\},\{a,b,c\}\}$. The six strict orders on $X$ are
\begin{equation*}
\succ_1:abc,\quad
\succ_2:acb,\quad
\succ_3:bac,\quad
\succ_4:bca,\quad
\succ_5:cab,\quad
\succ_6:cba,
\end{equation*}
where $abc$ means $a\succ b\succ c$. The analyst observes $\mu(\{a,b\})=\mu(\{a,b,c\})=1/2$ and the choice probabilities $\rho(a\mid\{a,b\})=1/10$, $\rho(a\mid\{a,b,c\})=9/10$, $\rho(c\mid\{a,b,c\})=0$ of the motivating example. Table~\ref{tab:app-cells} numbers the five observed cells and records, for each, its top cone and its mass $d(A,x)=\mu(A)\rho(x\mid A)$. These five masses are the only numerical inputs used in this appendix; later tables refer to cells by their number.
\begin{table}[ht]
\centering
\begin{tabular}{cccc}
\toprule
Cell no. & Cell $(A,x)$ & Top cone $\cT(x,A)$ & Mass $d(A,x)$\\
\midrule
1 & $(\{a,b\},a)$ & $\{\succ_1,\succ_2,\succ_5\}$ & $1/20$\\
2 & $(\{a,b\},b)$ & $\{\succ_3,\succ_4,\succ_6\}$ & $9/20$\\
3 & $(\{a,b,c\},a)$ & $\{\succ_1,\succ_2\}$ & $9/20$\\
4 & $(\{a,b,c\},b)$ & $\{\succ_3,\succ_4\}$ & $1/20$\\
5 & $(\{a,b,c\},c)$ & $\{\succ_5,\succ_6\}$ & $0$\\
\bottomrule
\end{tabular}
\caption{Observed cells, their top cones, and their masses.}
\label{tab:app-cells}
\end{table}

\subsection{Bounds from Theorem~\ref{thm:sharp}}\label{example:sharp}
\label{app:gale-identification}
By Theorem~\ref{thm:sharp}, a preference distribution $\nu$ is consistent with the data if and only if, for every collection of observed cells $\cE$,
\begin{equation}\label{eq:app-hall}
\sum_{(A,x)\in\cE}d(A,x)
\le
\nu\Big(
\bigcup_{(A,x)\in\cE}\cT(x,A)
\Big).
\end{equation}
This yields bounds on any cone mass $\nu(\cT(x,A))$ by two rules.
\emph{Lower bound:} a cell whose top cone is \emph{contained} in
$\cT(x,A)$ forces its mass into $\cT(x,A)$; this happens when the cell
records the same choice $x$ from a menu containing $A$. \emph{Upper
bound:} a cell whose top cone is \emph{disjoint} from $\cT(x,A)$ forces
its mass into the complement; this happens when the cell records a rival
$y\in A\setminus\{x\}$ chosen while $x$ was available. Summing the
disjoint cells bounds the complement from below, hence $\nu(\cT(x,A))$
from above.

To illustrate, take $\rs(a\mid\{a,b,c\})=\nu(\{\succ_1,\succ_2\})$. Cell 3 has top cone exactly $\{\succ_1,\succ_2\}$, so \eqref{eq:app-hall} with $\cE=\{3\}$ gives the lower bound $9/20$. Cells 2, 4, and 5 have top cones disjoint from $\{\succ_1,\succ_2\}$ (each records $b$ or $c$ chosen while $a$ was available), so \eqref{eq:app-hall} with $\cE=\{2,4,5\}$ gives $\nu(\cL\setminus\{\succ_1,\succ_2\})\ge 9/20+1/20+0=1/2$, hence the upper bound $1/2$. Table~\ref{tab:app-bounds} applies the same two rules to all five cone events; each entry can be checked against the top cones in Table~\ref{tab:app-cells} by inspection.
\begin{table}[ht]
\centering
\small
\begin{tabular}{lccccc}
\toprule
& \multicolumn{2}{c}{Cone $\subseteq\cT(x,A)$}
& \multicolumn{2}{c}{Cone $\cap\,\cT(x,A)=\emptyset$}
& \\
\cmidrule(lr){2-3}\cmidrule(lr){4-5}
$\rs(x\mid A)$ & Cells & Lower bound & Cells & Upper bound & Interval\\
\midrule
$\rs(a\mid\{a,b\})$   & 1, 3 & $1/2$  & 2, 4    & $1/2$ & $\{1/2\}$\\
$\rs(b\mid\{a,b\})$   & 2, 4 & $1/2$  & 1, 3    & $1/2$ & $\{1/2\}$\\
$\rs(a\mid\{a,b,c\})$ & 3    & $9/20$ & 2, 4, 5 & $1/2$ & $[9/20,\,1/2]$\\
$\rs(b\mid\{a,b,c\})$ & 4    & $1/20$ & 1, 3, 5 & $1/2$ & $[1/20,\,1/2]$\\
$\rs(c\mid\{a,b,c\})$ & 5    & $0$    & 3, 4    & $1/2$ & $[0,\,1/2]$\\
\bottomrule
\end{tabular}
\caption{Sharp bounds on the counterfactual choice probabilities. ``Lower bound'' is the total mass of the listed cells; ``Upper bound'' is one minus the total mass of the cells whose cones are disjoint from $\cT(x,A)$.}
\label{tab:app-bounds}
\end{table}

Note that on the binary menu $\{a,b\}$ the bounds collapse to a point: the lower bounds for $a$ and for $b$ are each $1/2$, and since $\cT(a,\{a,b\})$ and $\cT(b,\{a,b\})$ partition $\cL$, two lower bounds summing to one force both to hold with equality. The data thus point-identify $\rs(a\mid\{a,b\})=\rs(b\mid\{a,b\})=1/2$ with no assumption on the assignment rule. For the ternary menu $\{a,b,c\}$, by contrast, only intervals are identified, and the interval for $c$ has lower endpoint zero: cell 5 is empty, and the data cannot rule out that $c$-types exist but are steered away from $\{a,b,c\}$.

Steering bounds follow from Corollary~\ref{cor:steering-bounds}: for each cell with $d(A,x)>0$, the steering probability equals $d(A,x)/\nu(\cT(x,A))$, so its sharp interval is obtained by dividing the cell mass by the endpoints of the interval in Table~\ref{tab:app-bounds} (in reverse order). Table~\ref{tab:app-steering} reports the results; the benchmark under exogenous assignment is $\mu(A)=1/2$ for every cell.
\begin{table}[ht]
\centering
\begin{tabular}{cccc}
\toprule
Cell no. & $d(A,x)$ & Interval for $\rs(x\mid A)$ & Interval for $\pi(A\mid\cT(x,A))$\\
\midrule
1 & $1/20$ & $\{1/2\}$        & $\{1/10\}$\\
2 & $9/20$ & $\{1/2\}$        & $\{9/10\}$\\
3 & $9/20$ & $[9/20,\,1/2]$   & $[9/10,\,1]$\\
4 & $1/20$ & $[1/20,\,1/2]$   & $[1/10,\,1]$\\
\bottomrule
\end{tabular}
\caption{Sharp bounds on the steering probabilities.}
\label{tab:app-steering}
\end{table}

\subsection{Closed-Form Bounds from Theorem~\ref{thm:closedform}}\label{example:closed}
\label{app:closed-form-identification}
The bounds of Table~\ref{tab:app-bounds} required selecting, for each target, the right collection of cells. Theorem~\ref{thm:closedform} removes the selection. By part (ii) and the closed forms,
\begin{equation*}
\rs(x\mid A)
\in
\Bigg[
\sum_{B\in\cD:\,B\supseteq A}d(B,x),
\;\;
1-\sum_{B\in\cD:\,x\in B}\ \sum_{y\in(A\cap B)\setminus\{x\}}d(B,y)
\Bigg],
\end{equation*}
so each endpoint is determined by the following rule: the lower bound adds up the cells in which $x$ itself is chosen from a menu containing all of $A$; the upper bound subtracts the cells in which some rival $y\in A$ is chosen while $x$ is available. To illustrate, take $\rs(a\mid\{a,b,c\})$. The only observed menu containing $\{a,b,c\}$ is $\{a,b,c\}$ itself, so the lower bound is $d(\{a,b,c\},a)=9/20$. Both observed menus contain $a$, with rival sets $(A\cap B)\setminus\{a\}$ equal to $\{b\}$ for $B=\{a,b\}$ and $\{b,c\}$ for $B=\{a,b,c\}$, so the upper bound is $1-[d(\{a,b\},b)+d(\{a,b,c\},b)+d(\{a,b,c\},c)]=1-[9/20+1/20+0]=1/2$. Table~\ref{tab:app-closedform} records the same computation for all five targets. As proved in Theorem~\ref{thm:closedform}, the intervals coincide with those of Table~\ref{tab:app-bounds}.
\begin{table}[ht]
\centering
\small
\begin{tabular}{lccc}
\toprule
Target
& Lower bound
& Upper bound
& Interval\\
\midrule
$\rs(a\mid\{a,b\})$ & $\frac{1}{20}+\frac{9}{20}=\frac12$ & $1-\big(\frac{9}{20}+\frac{1}{20}\big)=\frac12$ & $\{\frac12\}$\\[3pt]
$\rs(b\mid\{a,b\})$ & $\frac{9}{20}+\frac{1}{20}=\frac12$ & $1-\big(\frac{1}{20}+\frac{9}{20}\big)=\frac12$ & $\{\frac12\}$\\[3pt]
$\rs(a\mid\{a,b,c\})$ & $\frac{9}{20}$ & $1-\big(\frac{9}{20}+\frac{1}{20}+0\big)=\frac12$ & $\big[\frac{9}{20},\frac12\big]$\\[3pt]
$\rs(b\mid\{a,b,c\})$ & $\frac{1}{20}$ & $1-\big(\frac{1}{20}+\frac{9}{20}+0\big)=\frac12$ & $\big[\frac{1}{20},\frac12\big]$\\[3pt]
$\rs(c\mid\{a,b,c\})$ & $0$ & $1-\big(\frac{9}{20}+\frac{1}{20}\big)=\frac12$ & $\big[0,\frac12\big]$\\
\bottomrule
\end{tabular}
\caption{The closed-form bounds of Theorem~\ref{thm:closedform}(ii). The lower bound is $\sum_{B\supseteq A}d(B,x)$; the upper bound is $1-\sum_{B:\,x\in B}\sum_{y\in(A\cap B)\setminus\{x\}}d(B,y)$.}
\label{tab:app-closedform}
\end{table}

Part (iii) then delivers the steering bounds by division. For instance, for the cell $(\{a,b,c\},a)$, with $d(\{a,b,c\},a)=9/20$,
\begin{equation*}
\pi(\{a,b,c\}\mid\cT(a,\{a,b,c\}))
\in
\left[
\frac{9/20}{1/2},
\;
\frac{9/20}{9/20}
\right]
=
\left[
\frac{9}{10},1
\right],
\end{equation*}
and the remaining cells reproduce Table~\ref{tab:app-steering} in the same way.

\subsection{Point Identification from Theorem~\ref{thm:pct}}\label{example:unique}
\label{app:pct-identification}
The intervals in Table~\ref{tab:app-bounds} are sharp: with a single market and no assumption on steering, they cannot be narrowed. This subsection uses the result in Section~\ref{sec:pct}, and shows that with common steering (Assumption~\ref{ass:pct-common}) and outcome-based steering (Assumption~\ref{ass:pct}), Theorem~\ref{thm:pct} recovers the exact counterfactuals and the exact assignment rule, given sufficient market variation. The exercise proceeds in three steps: I first specify the true primitives, then generate the observable cell masses from them, and finally show that an analyst who sees only the observables recovers the primitives uniquely.

\emph{True primitives.} Suppose the same platform serves $M=3$ markets. All primitives are collected in Table~\ref{tab:app-point}. Its first row is the steering rule, common to all markets and outcome-based: types predicted to choose $a$ from a feed see the augmented feed with probability $\gamma(\{a,b,c\},a)=9/10$ and the standard feed with probability $\gamma(\{a,b\},a)=1/10$; types predicted to choose $b$ face the reverse probabilities; and the rare $c$-types are treated like $a$-types, $\gamma(\{a,b,c\},c)=9/10$. Each supported type is assigned some menu with total probability one over the menus. The markets differ in tastes. Market $1$ is exactly the same as the motivating Example~\ref{sec:example}: $\nu_1$ places mass $\tfrac12$ on $\succ_1$ ($abc$) and $\tfrac12$ on $\succ_3$ ($bac$). The other two markets also contain types who would choose $c$: $\nu_2$ places mass $\left(\tfrac{3}{10},\tfrac{2}{5},\tfrac{3}{10}\right)$ and $\nu_3$ places mass $\left(\tfrac{1}{5},\tfrac{7}{10},\tfrac{1}{10}\right)$ on the orders $\left(\succ_1,\succ_3,\succ_5\right)=\left(abc,\,bac,\,cab\right)$. The remaining rows of Table~\ref{tab:app-point} are the implied counterfactuals, obtained by summing each $\nu_m$ over the top cones of Table~\ref{tab:app-cells}: $\rs_m(x\mid A)=\nu_m(\cT(x,A))$ is the total mass of the orders that rank $x$ first within $A$. For example, $\cT(a,\{a,b\})=\{\succ_1,\succ_2,\succ_5\}$, of which $\nu_2$ supports $\succ_1$ and $\succ_5$, so $\rs_2(a\mid\{a,b\})=\tfrac{3}{10}+\tfrac{3}{10}=\tfrac{3}{5}$; whereas $\cT(a,\{a,b,c\})=\{\succ_1,\succ_2\}$, so only $\succ_1$ contributes and $\rs_2(a\mid\{a,b,c\})=\tfrac{3}{10}$; and $\cT(c,\{a,b,c\})=\{\succ_5,\succ_6\}$, so $\rs_2(c\mid\{a,b,c\})=\nu_2(\succ_5)=\tfrac{3}{10}$. The remaining entries are computed in the same way.
\begin{table}[ht]
\centering
\begin{tabular}{lccccc}
\toprule
& \multicolumn{2}{c}{$A=\{a,b\}$} & \multicolumn{3}{c}{$A=\{a,b,c\}$}\\
\cmidrule(lr){2-3}\cmidrule(lr){4-6}
& $x=a$ & $x=b$ & $x=a$ & $x=b$ & $x=c$\\
\midrule
$\gamma(A,x)$        & $\frac{1}{10}$ & $\frac{9}{10}$ & $\frac{9}{10}$ & $\frac{1}{10}$ & $\frac{9}{10}$\\[2pt]
$\rs_1(x\mid A)$     & $\frac12$      & $\frac12$      & $\frac12$      & $\frac12$      & $0$\\[2pt]
$\rs_2(x\mid A)$     & $\frac35$      & $\frac25$      & $\frac{3}{10}$ & $\frac25$      & $\frac{3}{10}$\\[2pt]
$\rs_3(x\mid A)$     & $\frac{3}{10}$ & $\frac{7}{10}$ & $\frac15$      & $\frac{7}{10}$ & $\frac{1}{10}$\\
\bottomrule
\end{tabular}
\caption{True primitives: the steering rule and the market-level counterfactual choice probabilities.}
\label{tab:app-point}
\end{table}

\emph{Observables.} By \eqref{eq:collapse}, each observed cell mass is the product of a steering rule and a counterfactual, $d_m(A,x)=\gamma(A,x)\,\rs_m(x\mid A)$. For instance, $d_2(\{a,b,c\},a)=\tfrac{9}{10}\times\tfrac{3}{10}=\tfrac{27}{100}$. Multiplying Table~\ref{tab:app-point} out cell by cell gives Table~\ref{tab:app-markets}, the dataset the analyst observes; each row sums to one.
\begin{table}[ht]
\centering
\begin{tabular}{cccccc}
\toprule
& \multicolumn{5}{c}{Cell no.}\\
\cmidrule(lr){2-6}
Market & 1 & 2 & 3 & 4 & 5\\
\midrule
$m=1$ & $\frac{1}{20}$ & $\frac{9}{20}$ & $\frac{9}{20}$ & $\frac{1}{20}$ & $0$ \\[2pt]
$m=2$ & $\frac{3}{50}$ & $\frac{9}{25}$ & $\frac{27}{100}$ & $\frac{1}{25}$ & $\frac{27}{100}$ \\[2pt]
$m=3$ & $\frac{3}{100}$ & $\frac{63}{100}$ & $\frac{9}{50}$ & $\frac{7}{100}$ & $\frac{9}{100}$ \\
\bottomrule
\end{tabular}
\caption{Observed cell masses $d_m(A,x)$ in the three markets, by the cell numbering of Table~\ref{tab:app-cells}.}
\label{tab:app-markets}
\end{table}

\emph{Identification.} The analyst sees only Table~\ref{tab:app-markets}, and knows neither $\gamma$ nor the $\rs_m$. Recall the mechanics: for each menu $A$, the reciprocals of the steering rules solve the linear system \eqref{eq:linsys}, $D_A x=\1_M$ with $D_A=[d_m(A,x)]_{m,x}$, and the solution is unique when $\operatorname{rank}(D_A)=\left|A\right|$. Write $t=(t_a,t_b)$ for the solution on $\{a,b\}$ and $j=(j_a,j_b,j_c)$ for the solution on $\{a,b,c\}$.

For the menu $\{a,b\}$, the coefficient matrix consists of columns 1 and 2 of Table~\ref{tab:app-markets}:
\begin{equation*}
\begin{pmatrix}
\frac{1}{20} & \frac{9}{20}\\[2pt]
\frac{3}{50} & \frac{9}{25}\\[2pt]
\frac{3}{100} & \frac{63}{100}
\end{pmatrix}
\begin{pmatrix}
t_a\\
t_b
\end{pmatrix}
=
\begin{pmatrix}
1\\
1\\
1
\end{pmatrix}.
\end{equation*}
The two columns are not proportional, so $\operatorname{rank}(D_{\{a,b\}})=2$: the rank condition holds, and the unique solution is $t_a=10$, $t_b=10/9$. The steering rules are the reciprocals of the solution: $\gamma(\{a,b\},a)=1/t_a=1/10$. The counterfactuals are the observed cell masses divided by the steering rules. For instance,
\begin{equation*}
\rs_2(a\mid\{a,b\})
=
d_2(\{a,b\},a)\,t_a
=
\tfrac{3}{50}\times 10
=
\tfrac{3}{5}.
\end{equation*}

For the menu $\{a,b,c\}$, the coefficient matrix consists of columns 3--5 of Table~\ref{tab:app-markets}:
\begin{equation*}
\begin{pmatrix}
\frac{9}{20} & \frac{1}{20} & 0\\[2pt]
\frac{27}{100} & \frac{1}{25} & \frac{27}{100}\\[2pt]
\frac{9}{50} & \frac{7}{100} & \frac{9}{100}
\end{pmatrix}
\begin{pmatrix}
j_a\\
j_b\\
j_c
\end{pmatrix}
=
\begin{pmatrix}
1\\
1\\
1
\end{pmatrix}.
\end{equation*}
Direct computation confirms that the three columns are linearly independent, so $\operatorname{rank}(D_{\{a,b,c\}})=3$, and the unique solution is $j_a=10/9$, $j_b=10$, $j_c=10/9$. As before, $\gamma(\{a,b,c\},b)=1/j_b=1/10$, and, for instance, $\rs_3(c\mid\{a,b,c\})=d_3(\{a,b,c\},c)\,j_c=\tfrac{9}{100}\times\tfrac{10}{9}=\tfrac{1}{10}$. Repeating cell by cell recovers Table~\ref{tab:app-point}. If the population shares $\lambda_m$ are also observed, the population counterfactual follows by aggregation, $\rs(x\mid A)=\sum_m\lambda_m\rs_m(x\mid A)$.\footnote{With only two markets, identification holds on the binary menu
but not on the ternary one. The matrix $D_{\{a,b\}}$ is $2\times 2$ and
can satisfy the rank condition, in which case $\gamma(\{a,b\},\cdot)$ and
$\rs_m(\cdot\mid\{a,b\})$ in every market are point identified. The
matrix $D_{\{a,b,c\}}$, however, is $2\times 3$, so
$\operatorname{rank}(D_{\{a,b,c\}})\le 2<3$: its system admits a line of
solutions rather than a point, and neither $\gamma(\{a,b,c\},\cdot)$ nor
any $\rs_m(\cdot\mid\{a,b,c\})$ is point identified.}


\begin{thebibliography}{99}

\bibitem[Abaluck and Adams-Prassl(2021)]{AbaluckAdamsPrassl2021}
Abaluck, Jason, and Abi Adams-Prassl. 2021.
``What Do Consumers Consider Before They Choose? Identification from Asymmetric Demand Responses.''
\emph{Quarterly Journal of Economics}, 136(3): 1611--1663.

\bibitem[Agarwal and Somaini(2026)]{AgarwalSomaini2025}
Agarwal, Nikhil, and Paulo Somaini. 2026.
``Demand Analysis under Latent Choice Constraints.''
\emph{Review of Economic Studies}, 93(4), 2215--2249.

\bibitem[Aguiar and Kashaev(2025)]{AguiarKashaev2025}
Aguiar, Victor H., and Nail Kashaev. 2025.
``Identification and Estimation of Discrete Choice Models with Unobserved Choice Sets.''
\emph{Journal of Business \& Economic Statistics}, 43(1): 204--215.

\bibitem[Aguiar et al.(2023)]{AguiarEtAl2023}
Aguiar, Victor H., Maria Jose Boccardi, Nail Kashaev, and Jeongbin Kim. 2023.
``Random Utility and Limited Consideration.''
\emph{Quantitative Economics}, 14(1): 71--116.

\bibitem[Apesteguia, Ballester, and Lu(2017)]{ApesteguiaBallesterLu2017}
Apesteguia, Jose, Miguel A. Ballester, and Jay Lu. 2017.
``Single-Crossing Random Utility Models.''
\emph{Econometrica} 85: 661--674.

\bibitem[Apesteguia and Ballester(2023)]{ApesteguiaBallester2023}
Apesteguia, Jose, and Miguel A. Ballester. 2023.
``Random Utility Models with Ordered Types and Domains.''
\emph{Journal of Economic Theory} 211: 105674.

\bibitem[Barber{\`a} and Pattanaik(1986)]{BarberaPattanaik1986}
Barber{\`a}, Salvador, and Prasanta K. Pattanaik. 1986.
``Falmagne and the Rationalizability of Stochastic Choices in Terms of Random Orderings.''
\emph{Econometrica} 54: 707--715.

\bibitem[Barseghyan et al.(2021)]{BarseghyanEtAl2021}
Barseghyan, Levon, Maura Coughlin, Francesca Molinari, and Joshua C. Teitelbaum. 2021.
``Heterogeneous Choice Sets and Preferences.''
\emph{Econometrica}, 89(5): 2015--2048.

\bibitem[Block and Marschak(1960)]{BlockMarschak1960}
Block, H. D., and Jacob Marschak. 1960.
``Random Orderings and Stochastic Theories of Responses.''
In \emph{Contributions to Probability and Statistics: Essays in Honor of Harold Hotelling}, edited by I. Olkin, S. G. Ghurye, W. Hoeffding, W. G. Madow, and H. B. Mann, 97--132. Stanford: Stanford University Press.

\bibitem[Caplin and Dean(2015)]{CaplinDean2015}
Caplin, Andrew, and Mark Dean. 2015.
``Revealed Preference, Rational Inattention, and Costly Information Acquisition.''
\emph{American Economic Review}, 105(7): 2183--2203.

\bibitem[Caradonna and Turansick(2026)]{CaradonnaTuransick2026}
Caradonna, Peter P., and Christopher Turansick. 2026.
``Identification in Stochastic Choice.'' Working paper.

\bibitem[Cattaneo et al.(2020)]{CattaneoEtAl2020}
Cattaneo, Matias D., Xinwei Ma, Yusufcan Masatlioglu, and Elchin Suleymanov. 2020.
``A Random Attention Model.''
\emph{Journal of Political Economy}, 128(7): 2796--2836.

\bibitem[Chambers and Turansick(2025)]{ChambersTuransick2025}
Chambers, Christopher P., and Christopher Turansick. 2025.
``The Limits of Identification in Discrete Choice.''
\emph{Games and Economic Behavior} 150: 537--551.

\bibitem[Conlon and Mortimer(2013)]{ConlonMortimer2013}
Conlon, Christopher T., and Julie Holland Mortimer. 2013.
``Demand Estimation under Incomplete Product Availability.''
\emph{American Economic Journal: Microeconomics}, 5(4): 1--30.

\bibitem[Covington, Adams, and Sargin(2016)]{CovingtonAdamsSargin2016}
Covington, Paul, Jay Adams, and Emre Sargin. 2016.
``Deep Neural Networks for YouTube Recommendations.''
In \emph{Proceedings of the 10th ACM Conference on Recommender Systems},
191--198. New York: Association for Computing Machinery.

\bibitem[Crawford, Griffith, and Iaria(2021)]{CrawfordGriffithIaria2021}
Crawford, Gregory S., Rachel Griffith, and Alessandro Iaria. 2021.
``A Survey of Preference Estimation with Unobserved Choice Set Heterogeneity.''
\emph{Journal of Econometrics}, 222(1): 4--43.

\bibitem[Dardanoni et al.(2020)]{DardanoniEtAl2020}
Dardanoni, Valentino, Paola Manzini, Marco Mariotti, and Christopher J. Tyson. 2020.
``Inferring Cognitive Heterogeneity from Aggregate Choices.''
\emph{Econometrica}, 88(3): 1269--1296.

\bibitem[Dardanoni et al.(2023)]{DardanoniEtAl2023}
Dardanoni, Valentino, Paola Manzini, Marco Mariotti, Henrik Petri, and
Christopher J. Tyson. 2023.
``Mixture Choice Data: Revealing Preferences and Cognition.''
\emph{Journal of Political Economy}, 131(3): 687--715.

\bibitem[Falmagne(1978)]{Falmagne1978}
Falmagne, Jean-Claude. 1978.
``A Representation Theorem for Finite Random Scale Systems.''
\emph{Journal of Mathematical Psychology}, 18(1): 52--72.

\bibitem[Fiorini(2004)]{Fiorini2004}
Fiorini, Samuel. 2004.
``A Short Proof of a Theorem of Falmagne.''
\emph{Journal of Mathematical Psychology}, 48(1): 80--82.

\bibitem[Gale(1957)]{Gale1957}
Gale, David. 1957.
``A Theorem on Flows in Networks.''
\emph{Pacific Journal of Mathematics}, 7(2): 1073--1082.

\bibitem[Gaynor, Propper, and Seiler(2016)]{GaynorPropperSeiler2016}
Gaynor, Martin, Carol Propper, and Stephan Seiler. 2016.
``Free to Choose? Reform, Choice, and Consideration Sets in the English National Health Service.''
\emph{American Economic Review}, 106(11): 3521--3557.

\bibitem[Gibbard(2021)]{Gibbard2021}
Gibbard, Peter. 2021.
``Disentangling Preferences and Limited Attention: Random-Utility Models with Consideration Sets.''
\emph{Journal of Mathematical Economics}, 94: 102468.

\bibitem[Gul and Pesendorfer(2006)]{GulPesendorfer2006}
Gul, Faruk, and Wolfgang Pesendorfer. 2006.
``Random Expected Utility.''
\emph{Econometrica}, 74(1): 121--146.

\bibitem[Honka, Horta\c{c}su, and Vitorino(2017)]{HonkaHortacsuVitorino2017}
Honka, Elisabeth, Ali Horta\c{c}su, and Maria Ana Vitorino. 2017.
``Advertising, Consumer Awareness, and Choice: Evidence from the U.S. Banking Industry.''
\emph{RAND Journal of Economics}, 48(3): 611--646.

\bibitem[Kitamura and Stoye(2018)]{KitamuraStoye2018}
Kitamura, Yuichi, and J{\"o}rg Stoye. 2018.
``Nonparametric Analysis of Random Utility Models.''
\emph{Econometrica} 86: 1883--1909.

\bibitem[Knott, Hayes, and Neslin(2002)]{KnottHayesNeslin2002}
Knott, Aaron, Andrew Hayes, and Scott A. Neslin. 2002.
``Next-Product-to-Buy Models for Cross-Selling Applications.''
\emph{Journal of Interactive Marketing}, 16(3): 59--75.

\bibitem[Kono, Saito, and Sandroni(2025)]{KonoSaitoSandroni2025}
Kono, Haruki, Kota Saito, and Alec Sandroni. 2025.
``Random Utility with Unobservable Alternatives.''
\emph{American Economic Review}, forthcoming.

\bibitem[Liu and Luo(2025)]{LiuLuo2025}
Liu, Hui, and Yao Luo. 2025.
``Demand Analysis under Price Rigidity and Endogenous Assortment: An Application to China's Tobacco Industry.''
Working paper, arXiv:2501.17251.

\bibitem[Lops, de Gemmis, and Semeraro(2011)]{LopsDeGemmisSemeraro2011}
Lops, Pasquale, Marco de Gemmis, and Giovanni Semeraro. 2011.
``Content-Based Recommender Systems: State of the Art and Trends.''
In \emph{Recommender Systems Handbook}, edited by Francesco Ricci, Lior Rokach, Bracha Shapira, and Paul B. Kantor, 73--105. Boston: Springer.


\bibitem[Manski(2007)]{Manski2007}
Manski, Charles F. 2007.
``Partial Identification of Counterfactual Choice Probabilities.''
\emph{International Economic Review}, 48(4): 1393--1410.

\bibitem[Manzini and Mariotti(2014)]{ManziniMariotti2014}
Manzini, Paola, and Marco Mariotti. 2014.
``Stochastic Choice and Consideration Sets.''
\emph{Econometrica}, 82(3): 1153--1176.

\bibitem[Masatlioglu, Nakajima, and Ozbay(2012)]{MasatliogluNakajimaOzbay2012}
Masatlioglu, Yusufcan, Daisuke Nakajima, and Erkut Y. Ozbay. 2012.
``Revealed Attention.''
\emph{American Economic Review}, 102(5): 2183--2205.

\bibitem[McFadden and Richter(1991)]{McFaddenRichter1991}
McFadden, Daniel, and Marcel K. Richter. 1991.
``Stochastic Rationality and Revealed Stochastic Preference.''
In \emph{Preferences, Uncertainty and Rationality,}, edited by J. S. Chipman, D. McFadden, and M. K. Richter, 161--186. Boulder: Westview Press.

\bibitem[Shapley(1971)]{Shapley1971}
Shapley, Lloyd S. 1971.
``Cores of Convex Games.''
\emph{International Journal of Game Theory}, 1(1): 11--26.

\bibitem[Tomlinson, Ugander, and Benson(2021)]{TomlinsonUganderBenson2021}
Tomlinson, Kiran, Johan Ugander, and Austin R. Benson. 2021.
``Choice Set Confounding in Discrete Choice.''
In \emph{Proceedings of the 27th ACM SIGKDD Conference on Knowledge Discovery and Data Mining}. New York: Association for Computing Machinery, 1571--1581.

\bibitem[Turansick(2022)]{Turansick2022}
Turansick, Christopher. 2022.
``Identification in the Random Utility Model.''
\emph{Journal of Economic Theory} 203: 105489.

\end{thebibliography}
\end{document}